\documentclass[11pt]{article}
\usepackage{amsmath,amssymb,amsthm,fullpage,verbatim,bm,tikz}
\allowdisplaybreaks
\DeclareFontFamily{U}{mathx}{\hyphenchar\font45}
\DeclareFontShape{U}{mathx}{m}{n}{
      <5> <6> <7> <8> <9> <10>
      <10.95> <12> <14.4> <17.28> <20.74> <24.88>
      mathx10
      }{}
\DeclareSymbolFont{mathx}{U}{mathx}{m}{n}
\DeclareMathSymbol{\bigplus}{1}{mathx}{"90}
\DeclareMathSymbol{\bigtimes}{1}{mathx}{"91}

\DeclareRobustCommand{\stirling}{\genfrac[]{0pt}{}}

\newcommand*\circled[1]{\tikz[baseline=(char.base)]{
            \node[shape=circle,draw,inner sep=0.5pt] (char) {#1};}}
\DeclareMathOperator{\zigzag}{\circled{{\rm z}}}

\newtheorem{definition}{Definition}
\newtheorem{lemma}{Lemma}

\newtheorem{fact}{Fact}
\newtheorem{theorem}{Theorem}

\newcommand{\C}{\mathbb{C}}
\newcommand{\I}{\mathbb{I}}
\newcommand{\U}{\mathbb{U}}

\DeclareMathOperator*{\E}{{\rm {\bf E}}\,}
\DeclareMathOperator*{\Tr}{{\rm Tr}\;}

\DeclareMathOperator*{\spanning}{{\rm span}\;}

\newcommand{\zero}{\leavevmode\hbox{\small l\kern-3.5pt\normalsize0}}
\newcommand{\one}{\leavevmode\hbox{\small1\kern-3.8pt\normalsize1}}

\newcommand{\norm}[1]{\left\|{ #1 }\right\|}
\newcommand{\elltwo}[1]{\norm{#1}_2}

\newcommand{\ellinfty}[1]{\norm{#1}_\infty}

\newcommand{\inprod}[2]{\langle {#1}, #2 \rangle}

\newcommand{\polylog}{\mathrm{polylog}}
\newcommand{\poly}{\mathrm{poly}}

\newcommand{\cG}{\mathcal{G}}
\newcommand{\cH}{\mathcal{H}}
\newcommand{\cI}{\mathcal{I}}

\newcommand{\vcH}{\vec{\cH}}

\newcommand{\dU}{\dot{U}}
\newcommand{\dcG}{\dot{\cG}}
\newcommand{\ddcG}{\ddot{\cG}}

\title{{\bf Efficient quantum tensor product expanders and
unitary $t$-designs via the zigzag product
}}

\author{
Pranab Sen\thanks{
School of Technology and Computer Science, Tata Institute of Fundamental
Research, Mumbai 400005, India.
Email: {\sf pgdsen@tcs.tifr.res.in}
}
}

\date{}

\begin{document}
\maketitle

\begin{abstract}
A classical $t$-tensor product expander is a natural way of formalising
correlated walks of $t$ particles on a regular expander graph.
A quantum $t$-tensor product expander
is a completely positive trace preserving map that is a straightforward
analogue of a classical $t$-tensor product expander. Interest
in these maps arises from the fact that iterating a quantum $t$-tensor 
product expander gives us a unitary $t$-design, which has many applications
to quantum computation and information. We show that the zigzag
product of a high dimensional quantum expander (i.e. $t = 1$) of 
moderate degree
with a moderate dimensional quantum $t$-tensor product expander of
low degree gives us a high dimensional quantum $t$-tensor product expander 
of low degree. 
Previously such a result was known 
only for quantum expanders i.e. $t = 1$.

Using the zigzag product we give efficient constructions of 
quantum $t$-tensor product expanders in dimension $D$ where 
$t = \polylog(D)$. 
We then show
how replacing the zigzag product by the generalised zigzag product leads
to almost-Ramanujan quantum tensor product expanders i.e. having
near-optimal tradeoff between the degree $s$ and the second largest
singular value $s^{-\frac{1}{2} + O(\frac{1}{\sqrt{\log s}})}$. Both the 
products give better tradeoffs
between the degree and second largest singular value than what was
previously known for efficient constructions.
\end{abstract}

\section{Introduction}
Expander graphs are graphs of small degree and high connectivity, and
have had many applications to combinatorics and computer science
(see e.g. the survey paper \cite{HLW} and the references therein). One
way of formalising the {\em expansion property} of an infinite family
of directed graphs with out-degree and in-degree $d$ is via the 
requirement that the second 
largest singular value in absolute value of the normalised adjacency matrix
be at most $1 - \Omega(1)$. In this paper, we will use this notion
of an {\em algebraic expander}. An equivalent way to state the algebraic
definition is as follows: A $d$-regular expander on $D$ vertices with
second singular value $\lambda$ is
a linear transformation $\cG: \C^D \rightarrow \C^D$ that can be
expressed as
$
\cG(v) = \frac{1}{d} \sum_{i=1}^d P_i(v),
$
for any $v \in \C^D$, where
$\{P_i\}_{i=1}^d$ are $D \times D$ permutation matrices, such that
\[
\elltwo{\cG(v) - \frac{\vec{1}^T v}{D} \vec{1}} \leq \lambda \elltwo{v},
\]
$\vec{1}$ being the all ones $D$-tuple.
The above condition is equivalent to saying that
$
\ellinfty{\cG - \cI} \leq \lambda,
$
where $\ellinfty{M}$ is the largest singular value of $M$
aka Schatten $\ell_\infty$-norm of $M$ aka spectral norm of $M$, 
and $\cI$ is the `ideal' linear map defined by
$
\cI := \frac{1}{D!} \sum_{P \in S_D} P,
$
the average being over all $D \times D$ permutation matrices $P$.

The algebraic definition is directly used for many applications of
expander graphs. Ben-Aroya, Schwartz and Ta-Shma~\cite{BST} generalised 
the algebraic definition in a natural fashion to the quantum setting.
A $D$-dimensional quantum expander $\cG$ is a so-called 
completely positive trace
preserving superoperator, called CPTP map for short, which maps
$D \times D$ matrices to $D \times D$ matrices. If the input $D \times
D$ matrix is a density matrix, i.e. a Hermitian positive semidefinite
matrix with trace one, then the output is a density matrix also. The 
quantum expander $\cG$ is said to have degree $d$ if it can be
expressed as
$
\cG(M) = \frac{1}{d} \sum_{i=1}^d U_i M U_i^\dag,
$
where $M$ is the $D \times D$ input matrix and
$\{U_i\}_{i=1}^d$ are $D \times D$ unitary matrices.
Since matrices are the quantum analogue of vectors, and density matrices 
are the quantum analogue of probability 
distributions, a $D$-dimensional quantum expander of degree $d$ can 
be thought of a quantum analogue of a $d$-regular classical
expander on $D$ vertices. The quantum expander $\cG$ is said to have
second singular value at most $\lambda$ if 
\[
\elltwo{\cG(M) - \frac{\Tr [M]}{D} \one} \leq \lambda \elltwo{M},
\]
for all matrices $M \in \C^{D \times D}$, where $\one$ denotes the 
$D \times D$ identity matrix and $\elltwo{M}$ denotes
the Frobenius norm or the Schatten $\ell_2$-norm of $M$ which is nothing
but the $\ell_2$-norm of the $D^2$-tuple obtained by rearranging the
entries of matrix $M$. 
Equivalently, we can say that
$
\ellinfty{\cG - \cI} \leq \lambda,
$
where $\cI$ is the superoperator whose action on a $D \times D$ matrix
$M$ is defined by
$
\cI(M) := \int_{U \in \U(D)} U M U^\dag,
$
the integration being over the Haar probability measure on the 
unitary group $\U(D)$.

The concept of random walk of a single particle on a $d$-regular 
expander graph was naturally extended by Hastings and Harrow~\cite{HH}
to that of a {\em correlated} random walk of
$t$ particles. This walk can be informally 
described as follows: Toss a fair $d$-faced coin, if it comes up with
$i$ for some $1 \leq i \leq d$, then each of the $t$ particles takes
the $i$th outgoing edge from their respective vertices. Algebraically,
a classical $t$-tensor product expander ($t$-cTPE) on $D$ vertices of
degree $d$ can be defined as
a linear transformation 
$\cG: (\C^D)^{\otimes t} \rightarrow (\C^D)^{\otimes t}$ that can be
expressed as
$
\cG(v) = \frac{1}{d} \sum_{i=1}^d (P_i)^{\otimes t}(v),
$
for any $v \in (\C^D)^{\otimes t}$, where
$\{P_i\}_{i=1}^d$ are $D \times D$ permutation matrices. The 
$t$-tensor product expander is said to have second singular value
$\lambda$ if
$
\ellinfty{\cG - \cI} \leq \lambda,
$
where $\cI$ is the `ideal' linear map defined by
$
\cI := \frac{1}{D!} \sum_{P \in S_D} P^{\otimes t},
$
the average being over all $D \times D$ permutation matrices $P$.
 
In a similar fashion, Hastings and Harrow~\cite{HH} extended the notion
of a quantum expander to that of a quantum $t$-tensor product expander
($t$-qTPE).
A quantum $t$-tensor product expander ($t$-qTPE) in dimension $D$ of
degree $d$ can be defined as a CPTP superoperator
$
\cG: (\C^{D \times D})^{\otimes t} \rightarrow 
(\C^{D \times D})^{\otimes t}
$ 
that can be expressed as
$
\cG(M) = 
\frac{1}{d} \sum_{i=1}^d (U_i)^{\otimes t} M (U_i^\dag)^{\otimes t},
$
for any matrix $M \in (\C^{D \times D})^{\otimes t}$, where
$\{U_i\}_{i=1}^d$ are $D \times D$ unitary matrices. The 
qTPE is said to have second singular value $\lambda$ if
$
\ellinfty{\cG - \cI} \leq \lambda,
$
where $\cI$ is the `ideal' CPTP superoperator defined by its action on
a matrix $M$ by
$
\cI(M) := \int_{U \in \U(D)} U^{\otimes t} M (U^\dag)^{\otimes t},
$
the integration being over the Haar probability measure on the unitary
group $\U(D)$. A quantum expander defined above is thus a $1$-qTPE.

Quantum expanders have already found several applications in quantum
algorithms and complexity e.g.
\cite{AS, H, BST}. Quantum tensor product expanders are also of great
interest primarily because sequentially iterating a $t$-qTPE gives us
an approximate unitary $t$-design. Unitary $t$-designs have many
applications to quantum computation and information e.g. 
\cite{DCEL, AE, L}.
Many protocols in quantum information theory use Haar random unitaries
e.g. \cite{ADHW}.
Sometimes, these Haar random unitaries can be replaced by approximate
unitary $t$-designs which require less random bits to describe
e.g. \cite{SDTR}. Thus, unitary $t$-designs are a useful notion
of pseudo-random unitaries. More precisely, they serve as the quantum
analogue \cite{L} of $t$-wise independent random variables used in 
classical derandomisation applications (see e.g. \cite{BR}).
Moreover, Hastings and Harrow \cite{HH} have shown that $D^3$-qTPEs in 
dimension $D$ allow one to obtain better approximations to an arbitrary
$D \times D$ unitary matrix than what the Solovay-Kitaev theorem 
provides.

Obtaining efficient constructions of $t$-qTPEs is thus an important
problem. By an efficient construction, we mean that the qTPE 
superoperator in dimension $D$
can be realised by a quantum algorithm running in $\polylog(D)$ time.
An efficient construction of a $t$-qTPE will automatically give us
an efficient construction of an approximate unitary $t$-design simply
by sequential iteration. However, the converse is not known to be true.

Ben-Aroya, Schwartz and Ta-Shma~\cite{BST} showed 
that the zigzag product of expander graphs first defined by Reingold,
Vadhan and Wigderson~\cite{RVW} can be appropriately extended to the
quantum setting to give an efficient construction of a $1$-qTPE in
arbitrarily large dimension with constant degree and constant singular
value gap. This leaves open the case when $t > 1$. Efficient constructions
of approximate unitary $t$-designs for $t = \polylog(D)$ were known before
this work \cite{HL, BHH}, but viewed as expanders, they have 
$\polylog(D)$ degree. Moreover, the tradeoff between their degree
and singular value gap is far from optimal. This is unsatisfactory for 
some applications where we want constant degree and constant 
singular value gap. One such application is a quantum protocol for private
information retrieval via the quantum Johnson Lindenstrauss transform
\cite{S}. 

Hastings and Harrow \cite{HH} posed an open question asking whether the 
quantum expander constructions of
Ben-Aroya, Schwartz and Ta-Shma \cite{BST} can lead to quantum
tensor product expanders also. In this work, we answer their question
in the affirmative by showing that the zigzag product of a high 
dimensional quantum
expander i.e. $1$-qTPE with a low dimensional quantum tensor product 
expander 
i.e. $t$-qTPE gives rise to a high dimensional $t$-qTPE. Combined
with Hastings and Harrow's \cite{HH} existential result that random 
unitaries form a $t$-qTPE, we obtain the first efficient construction of
constant degree, constant singular value gap
$t$-qTPEs in arbitrarily large dimension. Our method achieves
the best known (in fact third power) tradeoff between degree and
singular value gap amongst efficient constructions of approximate
unitary $t$-designs. 

Proving that the zigzag product gives a $t$-qTPE is similar
to Ben-Aroya, Schwartz and Ta-Shma's proof \cite{BST} that the 
zigzag product 
gives a $1$-qTPE. However, we have to take
care of some technical geometric issues involving the eigenspace of
the qTPE superoperator for eigenvalue one. Unlike the $t=1$ setting,
this eigenspace has dimension larger than one which introduces
several complications. To address these complications, we define a
subspace that is `close' to the eigenspace in a certain precise sense.
The proof of closeness uses some combinatorial properties of permutations.
We then prove our main result by `switching back and forth' between these 
two spaces in a slightly tricky fashion. The formal proof can be
found in Section~\ref{sec:qzigzag}.

We then go further and achieve a near-optimal
tradefoff between degree $s$ and second largest singular 
value $s^{-\frac{1}{2} + O(\sqrt{\frac{1}{\log s}})}$ of a qTPE. 
The optimal tradeoff of
degree $s$ versus second largest singular value of $2s^{-1/2}$ is
known as the Ramanujan 
bound. To approach the Ramanujan bound, we need
to use the generalised zigzag product of graphs defined by
Ben-Aroya and Ta-Shma~\cite{BT}.
We can extend the generalised zigzag product to qTPEs
in a natural fashion. To show that the generalised zigzag product
gives an almost-Ramanujan qTPE we follow the outline of Ben-Aroya and
Ta-Shma's proof, combined with repeated applications of the 
`back and forth' technique explained
above. We also need to exploit a version of the Johnson-Lindenstrauss
property for an independent sequence of Haar random unitaries
(see e.g. \cite{S}).
This technical property can be viewed as the quantum generalisation
of the so-called {\em $\epsilon$-good} 
property of a sequence of independent
uniformly random permutations analysed in Ben-Aroya and Ta-Shma's paper
\cite[Lemma~19]{BT}. Though the quantum $\epsilon$-good property is
fundamentally different from the classical version in a certain sense,
nevertheless it allows us to prove that the generalised zigzag product
gives an almost-Ramanujan qTPE akin to the classical expander setting.

\section{Preliminaries}
\label{sec:preliminaries}
In this paper all vector spaces are over the field of complex numbers
$\C$, are finite dimensional and equipped
with inner products. Often we will be dealing with vector spaces whose
elements are matrices i.e. the elements are themselves linear operators
from a Hilbert space $A$ to a Hilbert space $B$. Such vector spaces
will be equipped with the {\em Hilbert-Schmidt inner product}
$\inprod{M}{N} := \Tr [M^\dagger N]$. This inner product is nothing
but the usual dot product of `long' vectors obtained by rearranging
the entries of matrices as tuples. We will 
also consider linear maps between
vector spaces that are themselves spaces of matrices. We will call
such linear maps as {\em superoperators}.

For vector space $\C^d$, let $e_i$, $i \in [d]$ denote the $i$th 
standard basis vector which consists of a one in position $i$ and zeroes
everywhere else.
Let $\elltwo{v} := \sqrt{\sum_{i=1}^d |v_i|^2}$ denote the 
$\ell_2$-norm of a vector $v \in \C^{d}$. The dot product of
two vectors is given by $\inprod{v}{w} := \sum_{i=1}^d v_i^* w_i$.
For a matrix $M \in \C^{d} \times \C^{d}$, let 
$\norm{M}_p$ denote its Schatten $p$-norm i.e. the $\ell_p$-norm of
the vector of singular values of $M$. We will only be interested in
Schatten $p$-norms with $p = 1, 2, \infty$. The Schatten $2$-norm
of $M$ turns out to be nothing but the $\ell_2$-norm of the long 
vector obtained by
rearranging the entries of the matrix $M$. In other words, it is the 
norm arising from the Hilbert-Schmidt inner product.

We will denote the vector space of $d \times d$ matrices, or equivalently
the vector space of linear maps $\C^d \rightarrow \C^d$, by 
$\C^{d \times d}$. In several places we will interchangeably use
$\C^{d^t}$ in place of the tensor product $(\C^d)^{\otimes t}$. In line
with this abuse of notation, we will sometimes use
$\C^{(Dd)^t \times (Dd)^t}$ to denote the vector space of linear maps
$
(\C^D \otimes \C^d)^{\otimes t} \rightarrow
(\C^D \otimes \C^d)^{\otimes t}.
$

For a Hilbert space $V$ and subspace $W \leq V$, we define the 
orthogonal complement of $W$ in $V$, denoted by $V \setminus W$, to
be the span of all vectors in $V$ orthogonal to $W$. When the ambient
space $V$ is clear from the context, we shall denote 
$V \setminus W$ by the shorter notation $W^\perp$.

\subsection{On permutations}
Let $t$, $d$ be positive integers. The following
lemma is easy to prove.
\begin{lemma}
\label{lem:perm}
Suppose $d \geq t$.
Define the falling factorial $(d)_t := d (d-1) \cdots (d-t+1)$. Then
$1 - \frac{(d)_t}{d!} \leq \frac{t(t-1)}{2d}$.
\end{lemma}
\begin{proof}
\[
1 - \frac{(d)_t}{d!} 
    =   
1 - 
(1-\frac{1}{d})(1-\frac{2}{d}) \cdots (1-\frac{t-1}{d}) 
  \leq  
\frac{1}{d}+\frac{2}{d} + \cdots + \frac{t-1}{d} 
  =  
\frac{t(t-1)}{2d}.
\]
\end{proof}

The number of permutations of $[t]$ with $k$ cycles, called
{\em (unsigned) Stirling number of the first kind}, is denoted by
$\stirling{t}{k}$. The unsigned Stirling numbers of the first 
kind satisfy the recurrence equation
$\stirling{t+1}{k} = t \stirling{t}{k} + \stirling{t}{k-1}$. 
From this, we can show by induction for $1 \leq k \leq t-1$ 
that $\stirling{t}{t-k} \leq {{t \choose 2} \choose k}$. 
This upper bound on 
$\stirling{t}{k}$ is almost tight by a result of
Arratia and DeSalvo~\cite[Theorem~3.2]{AD}.
Using this upper bound, we prove the following lemma.
\begin{lemma}
\label{lem:stirling}
Let $d$ be a positive integer larger than $t^2$. Let $M$ be
a $t! \times t!$-matrix whose rows and columns are
indexed by the permutations of $[t]$, defined as
follows:
\[
M_{\sigma \sigma'} :=
\begin{array}{ll}
d^{(t_{\sigma^{-1} \sigma'} - t)} & \sigma \neq \sigma' \\
0 & \sigma = \sigma',
\end{array}
\]
where $t_{\sigma^{-1} \sigma'}$ is the number of cycles in the
permutation $\sigma^{-1} \sigma'$. Then, $M$ is a real symmetric matrix
and $\ellinfty{M} \leq \frac{t(t-1)}{d}$. Moreover,
the eigenvalues of the real symmetric matrix $\one + M$ lie between
$1 - \frac{t(t-1)}{d}$ and $1 + \frac{t(t-1)}{d}$.
\end{lemma}
\begin{proof}
Observe that since $M$ is a real symmetric matrix, $\ellinfty{M}$ is
nothing but the largest eigenvalue of $M$ in absolute value. By
Gershgorin's theorem and the permutation symmetry of $M$, we conclude that
\[
\ellinfty{M} 
  \leq   
\sum_{\sigma \neq ()} d^{(t_\sigma - t)} 
  =  
\sum_{k=1}^{t-1} \stirling{t}{t-k} d^{-k} 
  \leq  
\sum_{k=1}^{t-1} {{t \choose 2} \choose k} d^{-k} 
  \leq  
\frac{{t \choose 2}}{d} \sum_{k=0}^{\infty} 2^{-k} 
  =
\frac{t(t-1)}{d},
\]
where $()$ denotes the identity permutation.
The claim about eigenvalues of $\one + M$ now follows easily.
This completes the proof.
\end{proof}

We will also need the following lemma.
\begin{lemma}
\label{lem:fixedpoints}
Fix $0 < \epsilon < \frac{1}{2t}$.
Let $N$ be
a $t! \times t!$-matrix whose rows and columns are
indexed by the permutations of $[t]$, defined as
follows:
\[
N_{\sigma \sigma'} :=
\begin{array}{ll}
\epsilon^{t - f_{\sigma^{-1} \sigma'}} & \sigma \neq \sigma' \\
0 & \sigma = \sigma',
\end{array}
\]
where $f_{\sigma^{-1} \sigma'}$ is the number of fixed points in the
permutation $\sigma^{-1} \sigma'$. Then, $N$ is a real symmetric matrix
and $\ellinfty{N} \leq 2 \epsilon^2 t^2 $.
\end{lemma}
\begin{proof}
Observe that since $N$ is a real symmetric matrix, $\ellinfty{N}$ is
nothing but the largest eigenvalue of $N$ in absolute value. By
Gershgorin's theorem and the permutation symmetry of $N$, we conclude that
\[
\ellinfty{N} 
  \leq   
\sum_{\sigma \neq ()} \epsilon^{t - f_\sigma} 
    =   
\sum_{k=1}^{t-2} \epsilon^{t-k} {t \choose k} (t-k)! 
(1 - \frac{1}{1!} + \frac{1}{2!} - \cdots + \frac{(-1)^t}{t!})
  \leq  
\sum_{k=1}^{t-2} \epsilon^{t-k} t^{t-k}
  \leq
2 (\epsilon t)^2,
\]
where $()$ denotes the identity permutation.
This completes the proof.
\end{proof}

\subsection{On Haar random unitaries}
\label{subsec:epsgood}
In this subsection, we single out a Johnson-Lindenstrauss type
of property of Haar random unitaries
that will be used in the proof that the generalised zigzag product gives
qTPEs.
Let $0 < \epsilon < 1$.  Let $V$, $V'$ be vector spaces of
dimensions $d$, $d'$.
Let $x$ be a unit length vector in
$V \otimes V'$.  Let $U$ be a unitary matrix on $V \otimes V'$.
Let $v$ be a computational basis vector for $V$.
The unitary $U$ is said to be {\em $\epsilon$-good} for $x$ given
$v_1$ if the probability of observing the outcome $v_1$ when the
system $V$ of the state $U x$ is measured in its computational basis
is $\frac{1 \pm 3 \epsilon}{d}$.
The unitary $U$ is said to be $\epsilon$-good for $x$ if for all
computational basis vectors $v_1 \in V$, $U$ is
$\epsilon$-good for $x$ given $v_1$.
The notation $U x | v_1$ denotes the normalised vector in $V \otimes V'$
obtained by computing $U x$, measuring only $V$ in its computational 
basis, and observing the result $v_1$.
Suppose $X$ is a set of unit length orthogonal vectors in 
$V \otimes V'$. 
The unitary $U$ is said to be $\epsilon$-good for $X$ if it is
$\epsilon$-good for each $x \in X$, and for every computational basis
vector $v_1 \in V$ and $x, x' \in X$, $\inprod{x}{x'} = 0$,
$
|\inprod{U x | v_1}{U x' | v_1}| \leq 8 \epsilon.
$

Let $(U_k, \ldots, U_1)$ be a $k$-tuple of unitaries on $V \otimes V'$.
Let $x_0 := e_{i_0} \otimes e'_{j_0}$ be a computational basis vector of 
$V \otimes V'$. Let $(e_{i_k}, \ldots, e_{i_1})$ be a $k$-tuple of
computational basis vectors of $V$. 
By induction on $j$, we define $x_j := U x_{j-1} | e_{i_j}$. 
We say that $U_j$ is
$\epsilon$-good given $(U_{j-1}, \ldots, U_1)$, 
$(e_{i_{j-1}}, \ldots, e_{i_1})$, $x_0$ if $U_j$ is $\epsilon$-good for
$x_{j-1}$.
The unitary $U_1$ is said to be $\epsilon$-good if it is $\epsilon$-good
for the set of all computational basis vectors $x_0$ of $V \otimes V'$.
We declare $U_j$ to be
$\epsilon$-good given $(U_{j-1}, \ldots, U_1)$ if $U_j$ is
$\epsilon$-good given $(U_{j-1}, \ldots, U_1)$, 
$(e_{i_{j-1}}, \ldots, e_{i_1})$, $x_0$ for the set of
all possible $(j-1)$-tuples of
computational basis vectors $(e_{i_{j-1}}, \ldots, e_{i_1})$ of $V$ and
all possible computational basis vectors $x_0$ of $V \otimes V'$. 
By induction on $j$,
we say that the $j$-tuple $(U_j, \ldots, U_1)$ is 
$\epsilon$-good if $(U_{j-1}, \ldots, U_1)$ is $\epsilon$-good and
$U_j$ is $\epsilon$-good given $(U_{j-1}, \ldots, U_1)$.

Let $\cH_j$, $1 \leq j \leq k$ be sets of unitaries on $V \otimes V'$,
each set $\cH_j$ being of size $s$. We use the shorthand
$\vcH$ to denote the $k$-tuple of sets $(\cH_k, \ldots, \cH_1)$. 
Then, $\vcH$ is said to be $\epsilon$-good if all $k$-tuples of
unitaries $(U_k, \ldots, U_1)$, $U_j \in \cH_j$ are $\epsilon$-good.

The following standard result is a version of the Johnson-Lindenstrauss
lemma for Haar random unitaries. It can be proved using 
Fact~2 and the method of Theorem~1 in \cite{S}.
\begin{fact}
\label{fact:epsgood}
Let $k$, $s$ be positive integers. Independently choose $sk$ Haar
random unitaries on $V \otimes V'$, and group them into $k$ sets
$\cH_j$, $1 \leq j \leq k$, each set $\cH_j$ being of size $s$. 
Then, the probability of $\vcH$ not being $\epsilon$-good is at most
$4 (s^{k+1} d^{k+2} d')^2 \exp(-2^{-4}\epsilon^2 d')$.
\end{fact}
\paragraph{Remark:}
The above fact can be thought of as a quantum analogue of
Lemma~19 in \cite{BT}, which analysed a similar property about an
independent sequence of uniformly random permutations. However, there 
is an important difference between the $\epsilon$-good properties
analysed in the two statemetns. The closeness to the uniform distribution
in the definition of the classical $\epsilon$-good property arises
from the choice of a uniformly random computational basis vector
of $V \otimes V'$.  In contrast, Fact~\ref{fact:epsgood} 
does not require us to choose a uniformly random computational
basis vector of $V \otimes V'$. In fact, Fact~\ref{fact:epsgood} works 
for any fixed
computational basis vector of $V \otimes V'$, whereas Lemma~19 of
\cite{BT} would give a deterministic result (which is the farthest
possible from the uniform distribution) if one were to take a
fixed computational basis vector of $V \otimes V'$.  This difference arises
from the inherently quantum effect of measurement being probabilistic.
Thus, the two statements are fundamentally different.

\subsection{Quantum tensor product expanders}
We recall the definition of  quantum tensor product 
expanders first defined by Hastings and Harrow~\cite{HH}.
\begin{definition}[{\bf Quantum tensor product expander}]
\label{def:qTPE}
A $(d, s, \lambda, t)$-quantum tensor product expander (qTPE) is a set of 
$d \times d$ unitaries $\{U_i\}_{i=1}^s$ such that
\[
\elltwo{
\E_U^{\mbox{Design}} [U^{\otimes t} M (U^\dagger)^{\otimes t}] -
\E_{U}^{\mbox{Haar}} [U^{\otimes t} M (U^\dagger)^{\otimes t}]
} \leq
\lambda \elltwo{M},
\]
for all linear operators 
$M: (\C^d)^{\otimes t} \rightarrow (\C^d)^{\otimes t}$.
The notation 
\[
\E_U^{\mbox{Design}} [U^{\otimes t} M (U^\dagger)^{\otimes t}] 
:=
s^{-1} \sum_{i=1}^s 
U_i^{\otimes t} M (U_i^\dagger)^{\otimes t} 
\]
denotes the expectation under the choice of a uniformly random unitary
from the design.
The notation 
\[
\E_U^{\mbox{Haar}} [U^{\otimes t} M (U^\dagger)^{\otimes t}] 
:=
\int_{\U(d)} U^{\otimes t} M (U^\dagger)^{\otimes t} \, d\mu
\]
denotes the expectation under 
the choice of a unitary $U$ picked from the Haar measure $\mu$ 
(formalisation of uniform measure) on the group of $d \times d$ unitary
matrices $\U(d)$.
\end{definition}
The quantity $s$ is referred to as the {\em degree} and
$1 - \lambda$ as the {\em singular value gap} of the qTPE. 

In this paper, we shall often consider Hermitian qTPEs. These are qTPEs
where the corresponding linear map on matrices, aka superoperator
$
M \mapsto
\E_U^{\mbox{Design}} [U^{\otimes t} M (U^\dagger)^{\otimes t}] 
$
is Hermitian under the Hilbert-Schmidt inner product of matrices. In fact, 
the Hermitian qTPEs constructed in this paper will be 
{\em explicitly Hermitian} i.e. there will be
a {\em bijective involution `-'} on the index
set $[s]$ satisfying $U_{-i} = U_i^{-1}$ for all $i \in [s]$. 
For a Hermitian qTPE, the singular value gap is also called the 
eigenvalue gap as the singular values of the associated Hermitian 
linear map are nothing but the absolute values of its eigenvalues.

Since the second superoperator in the definition of qTPE
is the Haar average
over a representation of the compact group $\U(d)$, it is 
equal to the orthogonal
projection onto the fixed space $W$ of the representation. 
Let $\sigma$ be a permutation of $[t]$.
Define the matrix
\[
\alpha_\sigma :=
d^{-t/2}
\sum_{(i_1, \ldots, i_t) \in [d]^t}
(e_{i_1} \otimes \cdots \otimes e_{i_t})
(
e_{i_{\sigma(1)}}^\dagger \otimes \cdots \otimes 
e_{i_{\sigma(t)}}^\dagger
).
\]
Thus $\alpha_\sigma$ is the operator obtained by first applying the 
Schatten $\ell_2$-normalised identity matrix in $\C^{d^t}$ followed
by shuffling the registers according to the permutation $\sigma$
i.e. register number $a$ goes to register number $\sigma(a)$. Let this
`shuffling' operator be denoted by $\Sigma^{(\C^d)^{\otimes t}}$.
Note that $\Sigma^{(\C^d)^{\otimes t}}$ is a unitary matrix.
Thus we have,
$
\alpha_\sigma = 
\Sigma^{(\C^d)^{\otimes t}} \frac{\one^{(\C^d)^{\otimes t}}}{d^{t/2}}.
$ 
Observe that for any $\sigma \in S_t$, $U \in \U(d)$, 
$U^{\otimes t} \alpha_\sigma = \alpha_\sigma U^{\otimes t}$. Thus it is clear
that $\alpha_\sigma$ lies in the fixed space $W$. It turns out that 
$W$ is spanned by the matrices $\alpha_\sigma$ as $\sigma$ ranges over all
permutations of $[t]$ i.e. 
$
W = \spanning_\sigma \alpha_\sigma.
$
This non-trivial statement follows from
Schur-Weyl duality \cite[Theorem 3.3.8]{GW}.
For $t \leq d$, it is easy to see that the matrices 
$\alpha_\sigma$, $\sigma \in S_t$ are linearly
independent and so $\dim W = t!$.

We thus see that for the superoperator 
$
\E_U^{\mbox{Design}} [U^{\otimes t} \cdot (U^\dagger)^{\otimes t}] -
\E_U^{\mbox{Haar}} [U^{\otimes t} \cdot (U^\dagger)^{\otimes t}],
$
the matrices $\{\alpha_\sigma\}_\sigma$ are left and right singular 
vectors
with singular value zero. Thus, the $\{\alpha_\sigma\}_\sigma$ are
eigenvectors with eigenvalue zero. This is true even if 
the design superoperator
is not Hermitian. The bound on the Schatten $\ell_\infty$-norm of the above
superoperator required by Definition~\ref{def:qTPE} translates to the
requirement that the other singular vectors have singular values 
at most $\lambda$. If the design superoperator is Hermitian, this is
equivalent to the requirement that the other eigenvectors have eigenvalues
between $-\lambda$ and $\lambda$.

Hastings and Harrow~\cite[Theorem 6]{HH} showed that independent Haar 
random choices of unitary matrices 
give rise to Hermitian TPEs with good tradeoff 
between degree and eigenvalue gap.
\begin{fact}[Random qTPE]
\label{fact:randomqTPE}
Let $s \geq 4$ be an even integer. Let $d \geq 100$ be an integer. 
Let $t$ be a positive integer satisfying 
$t \leq \frac{d^{1/6}}{10 \log d}$.
Choose $d \times d$ unitary matrices $\{U_i\}_{i=1}^{s/2}$ independently
from the Haar measure on the unitary group $\U(d)$.
For all $1 \leq i \leq s/2$, set $U_{i+\frac{s}{2}} := U_i^\dagger$. 
Then $\{U_i\}_{i=1}^s$ form an explicitly Hermitian 
$(d, s, \lambda, t)$-qTPE, where
$
\lambda < \frac{8}{s^{1/2}}
$
with probability at least $1 - d^{-2}$,
and the involution `-' is defined is $-i := i + \frac{s}{2}$ if 
$1 \leq i \leq s/2$, and $-i := i - \frac{s}{2}$ otherwise.
\end{fact}
Combining the above fact with Fact~\ref{fact:epsgood}, we get
\begin{lemma}
\label{lem:epsgood}
Let $s \geq 4$ be an even integer. Let $d \geq 100$ be an integer. 
Let $k \leq \log s$ be an integer. 
Let $0 < \epsilon < 10^{-2}$. 
Let $d' \geq 30 \log s (\log s + \log d) d^{2k+1} \epsilon^{-2}$.
Let $t$ be a positive integer satisfying 
$t \leq \frac{(dd')^{1/6}}{10 \log (dd')}$.
Choose $(dd') \times (dd')$ unitary matrices 
$\{U_i(j)\}$, $1 \leq i \leq s/2$, $1 \leq j \leq k$ independently
from the Haar measure on the unitary group $\U(dd')$.
For all $1 \leq i \leq s/2$, $1 \leq j \leq k$ 
set $U_{i+\frac{s}{2}}(j) := U_i^\dagger(j)$. 
Then, with probability at least $3/4$,  for all $j \in [k]$, 
$\cH_j := \{U_i(j)\}_{i=1}^s$ is an explicitly Hermitian 
$(d, s, \lambda, t)$-qTPE, where
$
\lambda < \frac{8}{s^{1/2}},
$
the involution `-' being defined as $-i := i + \frac{s}{2}$ if 
$1 \leq i \leq s/2$, and $-i := i - \frac{s}{2}$ otherwise, and
the $k$-tuple of expanders $\vcH$ is $\frac{\epsilon}{8 t d^{k}}$-good.
\end{lemma}

We now recall the definition of an approximate unitary $t$-design 
according to Low~\cite{L}.
\begin{definition}[{\bf Unitary $t$-design}]
Consider $d^2$ formal variables $\{u_{ij}\}_{i,j=1}^d$. A monomial $M$
in these formal variables is said to be {\em balanced of degree $t$} if
it is a product of exactly $t$ of the formal variables and exactly 
$t$ of complex conjugates of the formal variables (the sets
of unconjugated and conjugates variables bear no relation amongst them).
For a $d \times d$ unitary matrix $U$, let $M(U)$ denote the value
of the monomial $M$ obtained by evaluating it at the entries $U_{ij}$ of
$U$. A balanced polynomial of degree $t$ is a linear combination of
balanced monomials of degree $t$.

A unitary $(d, s, \alpha, t)$-design is a set of 
$d \times d$ unitaries $\{U_i\}_{i=1}^s$ such that
\[
\left|
\E_{U}^{\mbox{Design}} [M(U)] -
\E_{U}^{\mbox{Haar}} [M(U)]
\right| \leq
\frac{\alpha}{d^t},
\]
for all balanced monomials $M$ of degree $t$.
\end{definition}

The following facts are easy to prove from the definition of 
a qTPE and hold even if the qTPE is not Hermitian.
\begin{fact}
\label{fact:tracingout}
A $(d, s, \lambda, t)$-qTPE is also a $(d, s, \lambda, t')$-qTPE for
$t' \leq t$.
\end{fact}
\begin{fact}
\label{fact:doubling}
Suppose $\cG := \{U_i\}_{i=1}^s$ is a $(d, s, \lambda, t)$-qTPE. Then
$\cG' := \{U_i\}_{i=1}^s \cup \{U_i^\dagger\}_{i=1}^s$ is an explicitly
Hermitian  $(d, 2s, \lambda, t)$-qTPE with involution `-' defined by
$-i := i + s$ if 
$1 \leq i \leq s$, and $-i := i - s$ otherwise.
\end{fact}
\begin{fact}
\label{fact:squaring}
Let $\cH = \{U_j\}_{j=1}^s$ be a $(d, s, \lambda, t)$-qTPE.
Sequentially iterating $\cH$ twice means applying the superoperator 
corresponding to $\cH$ twice in succession. Then, $\cH \circ \cH$ is a
$(d, s^2, \lambda^2, t)$-qTPE where the $s^2$ unitaries are of the form
$U_i U_j$, $1 \leq i, j \leq s$. 
\end{fact}
For a balanced monomial 
$M = (u_{i_1 j_1} \cdots u_{i_t j_t}) (u_{i_1 j_1} \cdots u_{i_t j_t})^*$ 
of degree $t$, let $M'$ be the matrix with a one in the position
$((j_1, \ldots, j_t), (j_1, \ldots, j_t))$ and zeroes elsewhere. 
Plugging $M'$ into the definition of a  $(d, s, \lambda, t)$-qTPE, we
see that sequentially iterating the qTPE 
$O(\frac{t \log d + \log \alpha^{-1}}{\log \lambda^{-1}})$ times gives
us an $\alpha$-approximate unitary $t$-design. 
\begin{fact}
\label{fact:tensor}
Let $\cH = \{U_j\}_{j=1}^s$ be a $(d, s, \lambda, 1)$-qTPE.
The tensor product of $\cH$ with $\cH$ is defined as the
tensor product of the corresponding superoperators. Then,
$\cH \otimes \cH$ is a 
$(d^2, s^2, \lambda, 1)$-qTPE where the $s^2$ unitaries are of the form
$U_i \otimes U_j$, $1 \leq i, j \leq s$. 
\end{fact}

\section{Zigzag product gives a qTPE}
\label{sec:qzigzag}
Inspired by the definition of zigzag product for quantum expanders, i.e.
$1$-qTPEs, in \cite{BST}, 
we define the zigzag product of a $1$-qTPE and a $t$-qTPE as follows.
Let $\cG = \{U_i\}_{i=1}^d$ be a
$(D, d, \lambda_1, 1)$-qTPE and 
$\cH = \{V_j\}_{j=1}^s$ be a
$(d, s, \lambda_2, t)$-qTPE. We will also use $\cG$, $\cH$ to denote the
corresponding  superoperators
$\C^{D \times D} \rightarrow \C^{D \times D}$,
$(\C^{d \times d})^{\otimes t} \rightarrow (\C^{d \times d})^{\otimes t}$. 
If $\cG$ is explicitly Hermitian, let `$-$' denote the corresponding
involution on $[d]$; if not, let `$-$' denote the identity function
on $[d]$.
Define the unitary matrix $\dcG$ on the vector space 
$\C^{Dd} \cong \C^{D} \otimes \C^d$
by
\[
e_{a} \otimes e_{b} \stackrel{\dcG}{\mapsto}
(U_{b} e_{a}) \otimes e_{-b},
\]
where $e_{a}$, $e_{b}$ denote computational basis
vectors of $\C^D$, $\C^d$ respectively. If $\cG$ is explicitly
Hermitian, then $\dcG$ is an involution
i.e. $\dcG^2 = \one^{\C^{Dd}}$. In this case, $\dcG$ is both unitary 
and Hermitian.
Let $\one^{\C^D}$ denote the identity operator on $\C^{D}$.
We define the unitary superoperator 
$\ddcG$ on $\C^{(Dd) \times (Dd)}$ as
\[
M \stackrel{\ddcG}{\mapsto} \dcG M \dcG^{-1}.
\]
If $\cG$ is explicitly Hermitian, then $\ddcG$ is Hermitian also.
\begin{definition}[Zigzag product of qTPEs]
\label{def:qzigzag}
The zigzag product of qTPEs $\cG$ and $\cH$, denoted by
$\cG \zigzag \cH$, is defined as the following set
of $s^2$ unitary matrices on $\C^{Dd}$:
\[
\cG \zigzag \cH :=
\{
(\one^{\C^D} \otimes V_i) \dcG
(\one^{\C^D} \otimes V_j):
i, j \in [s]
\}.
\]
\end{definition}
Let $\I^{\C^{D^t \times D^t}}$
denote the identity superoperator on $\C^{D^t \times D^t}$.
Viewed as a superoperator on 
$
(\C^{D \times D} \otimes \C^{d \times d})^{\otimes t} \cong
\C^{D^t \times D^t} \otimes \C^{d^t \times d^t}  \cong
(\C^{(Dd) \times (Dd)})^{\otimes t} \cong
\C^{(Dd)^t \times (Dd)^t},
$ 
the zigzag product $\cG \zigzag \cH$ is nothing but
\[
\cG \zigzag \cH :=
(\I^{\C^{D^t \times D^t}} \otimes \cH) \circ 
\ddcG^{\otimes t} \circ
(\I^{\C^{D^t \times D^t}} \otimes \cH).
\]
If $\cG$ and $\cH$ are explicitly Hermitian, then
$\cG \zigzag \cH$ is also explicitly Hermitian with 
involution $-(i,j) := (-j, -i)$ on $[s^2] \cong [s] \times [s]$.

Suppose $t \leq d \leq D \leq Dd$. Then,
the eigenspace $W$ of the superoperator $\cG \zigzag \cH$ for
eigenvalue $1$ is spanned by the linearly independent matrices
$\{\alpha_\sigma\}_{\sigma \in S_t}$ where 
\[
\C^{(Dd)^t \times (Dd)^t} \ni \alpha_\sigma := 
\Sigma^{(\C^D \otimes \C^d)^{\otimes t}} 
\frac{\one^{(\C^D \otimes \C^d)^{\otimes t}}}{(Dd)^{t/2}} =
(\alpha_1)_\sigma \otimes (\alpha_2)_\sigma,
\]
\[
(\alpha_1)_\sigma := 
\Sigma^{(\C^D)^{\otimes t}} 
\frac{\one^{(\C^D)^{\otimes t}}}{D^{t/2}} \in \C^{D^t \times D^t},
~~~
(\alpha_2)_\sigma := 
\Sigma^{(\C^d)^{\otimes t}}
\frac{\one^{(\C^d)^{\otimes t}})}{d^{t/2}} \in \C^{d^t \times d^t}.
\]
Define the matrices
\begin{eqnarray*}
\C^{D^t \times D^t} \ni \alpha_1 
& := &
(\alpha_1)_{()} =
D^{-t/2} \one^{(\C^D)^{\otimes t}}, \\
\C^{d^t \times d^t} \ni \alpha_2 
& := &
(\alpha_2)_{()} =
d^{-t/2} \one^{(\C^d)^{\otimes t}}, \\
\C^{d^t \times d^t} \ni \alpha'_2 
& := &
d^{-t/2} \sum_{(j_1, \ldots, j_t) \in [d], \mbox{distinct}}
(e_{j_1} \otimes \cdots \otimes e_{j_t})
(e_{j_1}^\dagger \otimes \cdots \otimes e_{j_t}^\dagger), \\
\C^{(Dd)^t \times (Dd)^t} \ni \alpha' 
& := &
\alpha_1 \otimes \alpha'_2.
\end{eqnarray*}
For $\sigma \in S_t$, define
\[
\C^{d^t \times d^t} \ni (\alpha'_2)_\sigma := 
\Sigma^{(\C^d)^{\otimes t}} \alpha'_2,
~~~ 
\C^{(Dd)^t \times (Dd)^t} \ni \alpha'_\sigma := 
(\alpha_1)_\sigma \otimes (\alpha'_2)_\sigma.
\]
Observe that the respective sets of matrices 
$\{(\alpha)_\sigma\}_{\sigma \in S_t}$,
$\{(\alpha'_2)_\sigma\}_{\sigma \in S_t}$ are orthogonal under the 
Hilbert-Schmidt inner product.

Define the vector spaces
\\
\begin{minipage}[t]{0.45\textwidth}
\begin{eqnarray*}
W_1 
& := & 
\spanning \{(\alpha_1)_\sigma\}_{\sigma \in S_t}, \\
W_2 
& := & 
\spanning \{(\alpha_2)_\sigma\}_{\sigma \in S_t}, \\
W'_2 
& := & 
\spanning \{(\alpha'_2)_\sigma\}_{\sigma \in S_t}, \\
W' 
& := & 
\spanning \{\alpha'_\sigma\}_{\sigma \in S_t}, 
\end{eqnarray*}
\end{minipage}
~
\begin{minipage}[t]{0.45\textwidth}
\begin{equation}
\label{eq:subspaces}
\begin{array}{rcl}
W^\perp
& := & 
\C^{(Dd)^t \times (Dd)^t} \setminus W, \\
(W')^\perp
& := & 
\C^{(Dd)^t \times (Dd)^t} \setminus W', \\
(W_1)^\perp
& := & 
\C^{D^t \times D^t} \setminus W_1, \\
(W'_2)^\perp
& := & 
\C^{d^t \times d^t} \setminus W'_2. 
\end{array}
\end{equation}
\end{minipage}\\
Then 
\[
(W')^\perp = 
(\C^{D^t \times D^t} \otimes (W'_2)^\perp) \oplus
((\C^{D^t \times D^t} \otimes W'_2) \cap (W')^\perp).
\]

We now define a geometric property called {\em everywhere close} capturing
that a subspace $W$ is `close' in a certain sense to a subspace $W'$.
\begin{definition}
\label{def:everywhereclose}
Let $W$, $W'$ be subspaces of a vector space $V$. Let $\epsilon > 0$. 
We say that $W$ is
everywhere close to $W'$ within $\epsilon$ if for all $w \in W$,
$\elltwo{w} = 1$, there is a $w' \in W'$ such that 
$\elltwo{w - w'} \leq \epsilon$. If $W$ is everywhere close to $W'$
within $\epsilon$
and $W'$ is everywhere close to $W$ within $\epsilon$, then we say that
subspaces $W$, $W'$ are everywhere close within $\epsilon$.
\end{definition}

We next prove an important property about the subspaces $W$, $W'$,
$\C^{D^t \times D^t} \otimes W_2$, $\C^{D^t \times D^t} \otimes W'_2$
and their orthogonal spaces defined in Equation~\ref{eq:subspaces} above.
\begin{lemma}
\label{lem:subspaces}
For the definitions of the subspaces above, the following claims are
true.
\begin{enumerate}

\item
The subspaces $W$, $W'$ are everywhere close to within
$2 \sqrt{\frac{t(t-1)}{d}}$.

\item
For any
$
\beta' = 
\sum_\sigma b_\sigma (\beta'_1)_\sigma \otimes (\alpha'_2)_\sigma 
\in \C^{D^t \times D^t} \otimes W'_2, 
$
$\elltwo{\beta'} = 1$,
define
$
\beta := 
\sum_\sigma b_\sigma (\beta'_1)_\sigma \otimes (\alpha_2)_\sigma 
\in \C^{D^t \times D^t} \otimes W_2. 
$
Then $\elltwo{\beta - \beta'} \leq 2 \sqrt{\frac{t(t-1)}{d}}$.
In particular,
the subspace $\C^{D^t \times D^t} \otimes W'_2$ is everywhere close to
the subspace $\C^{D^t \times D^t} \otimes W_2$ to within
$2 \sqrt{\frac{t(t-1)}{d}}$. 

\item
The subspaces $W^\perp$, $(W')^\perp$ are everywhere close to within
$2 \sqrt[4]{\frac{t(t-1)}{d}}$.

\item
The subspace $\C^{D^t \times D^t} \otimes (W'_2)^\perp$ is
everywhere close to the subspace 
$\C^{D^t \times D^t} \otimes (W_2)^\perp$ to within
$2 \sqrt[4]{\frac{t(t-1)}{d}}$.

\end{enumerate}
\end{lemma}
\begin{proof}
Observe that for permutations $\sigma, \sigma' \in S_t$,
\[
\inprod{(\alpha_2)_{\sigma'}}{(\alpha_2)_{\sigma}} =
d^{-t} \Tr[((\Sigma')^{(\C^d)^{\otimes t}})^\dagger 
\Sigma^{(\C^d)^{\otimes t}}
] =
d^{(t_{(\sigma')^\dagger \sigma} - t)},
\]
where $t_{(\sigma')^\dagger \sigma}$ is the number of cycles in the
permutation $(\sigma')^\dagger \sigma$. Similarly,
$
\inprod{(\alpha_1)_{\sigma'}}{(\alpha_1)_{\sigma}} =
D^{(t_{(\sigma')^\dagger \sigma} - t)}.
$

Let $\beta \in W$, $\elltwo{\beta} = 1$. Express $\beta$ as a linear
combination
$
\beta = \sum_\sigma a_\sigma (\alpha_1)_\sigma \otimes (\alpha_2)_\sigma,
$
where $a_\sigma \in \C$. 
We have,
\begin{eqnarray*}
1 
& = &
\inprod{\beta}{\beta} 
\;=\;
\sum_{\sigma', \sigma} a_{\sigma'}^* a_\sigma
\inprod{(\alpha_1)_{\sigma'} \otimes (\alpha_2)_{\sigma'}}
{(\alpha_1)_\sigma \otimes (\alpha_2)_\sigma} \\
& = &
\sum_{\sigma', \sigma} a_{\sigma'}^* a_\sigma
\inprod{(\alpha_1)_{\sigma'}}{(\alpha_1)_\sigma} \cdot
\inprod{(\alpha_2)_{\sigma'}}{(\alpha_2)_\sigma}
\;  = \;
\sum_{\sigma', \sigma} a_{\sigma'}^* a_\sigma
(Dd)^{(t_{(\sigma')^\dagger \sigma}-t)} \\
&  =   &
a^\dagger (\one + M) a 
\;\in\;
\elltwo{a}^2 \left(1 \pm \frac{t(t-1)}{Dd}\right),
\end{eqnarray*}
where $a$ is a $t!$-tuple whose $\sigma$th entry is $a_\sigma$,
and $M$ is the $t! \times t!$-matrix
defined in Lemma~\ref{lem:stirling} with $Dd$ replacing $d$. 
We thus conclude that
$
\sum_\sigma |a_\sigma|^2 \in \frac{1}{1 \pm \frac{t(t-1)}{Dd}}.
$

For $\beta \in W$, $\elltwo{\beta} = 1$ as defined above, let
$
\beta' := 
\sum_\sigma a_\sigma (\alpha_1)_\sigma \otimes (\alpha'_2)_\sigma 
\in W'.
$
Then,
$
\elltwo{\beta'}^2 = 
\sum_\sigma 
|a_\sigma|^2 
\elltwo{(\alpha'_1)_\sigma}^2 \elltwo{(\alpha'_2)_\sigma}^2 =
\sum_\sigma 
|a_\sigma|^2 
\elltwo{(\alpha'_2)_\sigma}^2.
$
By Lemma~\ref{lem:perm} and the previous paragraph, we get
\[
\frac{1 - \frac{t(t-1)}{2d}}{1 + \frac{t(t-1)}{Dd}} \leq
\elltwo{\beta'}^2 \leq
\frac{1}{1 - \frac{t(t-1)}{Dd}}.
\]
Observe that if $\sigma \neq \sigma'$, 
$
\inprod{(a'_2)_{\sigma'}}{(a_2)_\sigma} =
0.
$
Moreover,
$
\inprod{(a'_2)_{\sigma}}{(a_2)_\sigma} \geq
1 - \frac{t(t-1)}{2d}.
$
Thus,
\begin{eqnarray*}
\inprod{\beta'}{\beta} 
& = &
\sum_{\sigma', \sigma} a_{\sigma'}^* a_\sigma
\inprod{(\alpha_1)_{\sigma'} \otimes (\alpha'_2)_{\sigma'}}
{(\alpha_1)_\sigma \otimes (\alpha_2)_\sigma} \\
& = &
\sum_{\sigma} |a_\sigma|^2
\inprod{(\alpha_1)_{\sigma}}{(\alpha_1)_\sigma}
\inprod{(\alpha'_2)_{\sigma}}{(\alpha_2)_\sigma} \\
& \geq &
\frac{1 - \frac{t(t-1)}{2d}}{1 + \frac{t(t-1)}{Dd}} 
\;\geq\;
1 - \frac{t(t-1)}{d}.
\end{eqnarray*}
Hence,
\[
\elltwo{\beta - \beta'}^2 =
\elltwo{\beta}^2 + \elltwo{\beta'}^2 - 
2 \Re (\inprod{\beta'}{\beta}) \leq
1 + \frac{1}{1 - \frac{t(t-1)}{Dd}} - 
2\left(1 - \frac{t(t-1)}{d}\right) \leq
\frac{4 t(t-1)}{d}.
\]
This shows that for any $\beta \in W$, $\elltwo{\beta} = 1$, there 
exists a $\beta' \in W'$ such that 
$\elltwo{\beta - \beta'} \leq 2\sqrt{\frac{t(t-1)}{d}}$. 

Similarly, we can show that for any 
$\beta' \in \C^{D^t \times D^t} \otimes W'_2$, 
$\elltwo{\beta'} = 1$, there exists a 
$\beta \in \C^{D^t \times D^t} \otimes W_2$ such that
$\elltwo{\beta - \beta'} \leq 2\sqrt{\frac{t(t-1)}{d}}$.
Moreover, if $\beta' \in W'$, then the resulting $\beta \in W$.
We proceed as follows. Suppose
$
\beta' := 
\sum_\sigma b_\sigma (\beta_1)_\sigma \otimes (\alpha'_2)_\sigma, 
$
where $b_\sigma \in \C$, $(\beta_1)_\sigma \in \C^{D^t \times D^t}$,
$\elltwo{(\beta_1)_\sigma} = 1$. Then by Lemma~\ref{lem:perm},
\[
1 =
\elltwo{\beta'}^2 = 
\sum_\sigma |b_\sigma|^2 \elltwo{(\alpha'_2)_\sigma}^2 \geq
\left(1-\frac{t(t-1)}{2d}\right)
\sum_\sigma |b_\sigma|^2,
\]
which implies that
$
\sum_\sigma |b_\sigma|^2 \leq
\frac{1}{1-\frac{t(t-1)}{2d}}.
$
Similarly, we can argue that
$
\sum_\sigma |b_\sigma|^2 \geq 1.
$
Define
$
\beta := 
\sum_\sigma b_\sigma (\beta_1)_\sigma \otimes (\alpha_2)_\sigma. 
$
Then,
\begin{eqnarray*}
\elltwo{\beta}^2
& = &
|\sum_{\sigma', \sigma} b_{\sigma'}^* b_\sigma
\inprod{(\beta_1)_{\sigma'} \otimes (\alpha_2)_{\sigma'}}
{(\beta_1)_\sigma \otimes (\alpha_2)_\sigma}| \\
& = &
|\sum_{\sigma', \sigma} b_{\sigma'}^* b_\sigma
\inprod{(\beta_1)_{\sigma'}}{(\beta_1)_\sigma} \cdot
\inprod{(\alpha_2)_{\sigma'}}{(\alpha_2)_\sigma}| \\
& \leq &
\sum_{\sigma', \sigma} |b_{\sigma'}| |b_\sigma|
|\inprod{(\beta_1)_{\sigma'}}{(\beta_1)_\sigma}| \cdot
|\inprod{(\alpha_2)_{\sigma'}}{(\alpha_2)_\sigma}| \\
& \leq &
\sum_{\sigma', \sigma} |b_{\sigma'}| |b_\sigma|
d^{(t_{(\sigma')^\dagger \sigma} - t)}
\; = \;
b^\dagger (\one + M) b 
\;\leq\;
\elltwo{b}^2 \ellinfty{\one + M}
\;\leq\;
\frac{1 + \frac{t(t-1)}{d}}{1-\frac{t(t-1)}{2d}},
\end{eqnarray*}
where $b$ is a $t!$-tuple whose $\sigma$th entry is $|b_\sigma|$,
and $M$ is the $t! \times t!$-matrix
defined in Lemma~\ref{lem:stirling}. Hence,
\begin{eqnarray*}
\inprod{\beta'}{\beta} 
& = &
\sum_{\sigma', \sigma} b_{\sigma'}^* b_\sigma
\inprod{(\beta_1)_{\sigma'} \otimes (\alpha'_2)_{\sigma'}}
{(\beta_1)_\sigma \otimes (\alpha_2)_\sigma} \\
& = &
\sum_{\sigma} |b_\sigma|^2
\inprod{(\beta_1)_{\sigma}}{(\beta_1)_\sigma}
\inprod{(\alpha'_2)_{\sigma}}{(\alpha_2)_\sigma} \\
& \geq &
1 - \frac{t(t-1)}{2d}.
\end{eqnarray*}
Hence,
\[
\elltwo{\beta - \beta'}^2 =
\elltwo{\beta}^2 + \elltwo{\beta'}^2 - 
2 \Re (\inprod{\beta'}{\beta}) \leq
\frac{1 + \frac{t(t-1)}{d}}{1-\frac{t(t-1)}{2d}} + 1
- 2\left(1 - \frac{t(t-1)}{2d}\right) \leq
\frac{3 t(t-1)}{d}.
\]
This proves the first two claims of the lemma.

Using the first claim of the lemma, we 
now argue that $(W')^\perp$ is everywhere close to $W^\perp$ 
to within $2 \sqrt[4]{\frac{t(t-1)}{d}}$.
Let $\gamma' \in (W')^\perp$, $\elltwo{\gamma'} = 1$. 
Let $\beta \in W$ be the projection of $\gamma'$ onto $W$. 
We claim that $\elltwo{\beta} \leq 2\sqrt[4]{\frac{t(t-1)}{d}}$. Suppose 
not.  Then 
$\elltwo{\gamma' - \beta} \leq \sqrt{1 - 4\sqrt{\frac{(t-1)}{d}}}$. 
Since $W'$ and $W$ are everywhere close, 
there exists a $\beta' \in W'$ such that
\[
\elltwo{\beta - \beta'} \leq \elltwo{\beta} 2\sqrt{\frac{t(t-1)}{d}} \leq
2 \sqrt{\frac{t(t-1)}{d}}.
\]
Hence,
$
\elltwo{\gamma' - \beta'} \leq 
\sqrt{1 - 4\sqrt{\frac{t(t-1)}{d}}} + 2\sqrt{\frac{t(t-1)}{d}} 
< 1,
$
leading to a contradiction because the orthogonality of $\gamma'$ 
and $\beta'$ would imply $\elltwo{\gamma' - \beta'} \geq 1$. Define 
$\gamma := \gamma' - \beta$. Then $\gamma \in W^\perp$ and
$
\elltwo{\gamma' - \gamma} = \elltwo{\beta} \leq
2 \sqrt[4]{\frac{t(t-1)}{d}}.
$
This shows that $(W')^\perp$ is everywhere close to $W^\perp$ to within
$2 \sqrt[4]{\frac{t(t-1)}{d}}$.
The argument can be reversed to also show that
$W^\perp$ is everywhere close to $(W')^\perp$ to within
$2 \sqrt[4]{\frac{t(t-1)}{d}}$.
We have thus finished showing that $W^\perp$ and $(W')^\perp$ are
everywhere close to within 
$2 \sqrt[4]{\frac{t(t-1)}{d}}$.

By a similar argument, starting from the second claim of the lemma,
we can prove the fourth claim.
\end{proof}

We now prove the following theorem.
\begin{theorem}
Let $\cG = \{U_i\}_{i=1}^d$ be a
$(D, d, \lambda_1, 1)$-qTPE and 
$\cH = \{V_j\}_{j=1}^s$ be a
$(d, s, \lambda_2, t)$-qTPE, where $D \geq d \geq 10 t^2$.  
Then $\cG \zigzag \cH$ is a
$(Dd, s^2, \lambda, t)$-qTPE where 
\[
\lambda := 
\lambda_1 + \lambda_2 + \lambda_2^2 + 24\sqrt[4]{\frac{t(t-1)}{d}}.
\]
\end{theorem}
\begin{proof}
In order to prove the theorem,
we need to show that for any matrix $\gamma \in W^\perp$ with 
$\elltwo{\gamma} = 1$, 
$
\elltwo{(\cG \zigzag \cH)(\gamma)} \leq \lambda.
$
By Lemma~\ref{lem:subspaces},
$W^\perp$ and $(W')^\perp$ are everywhere close. 
Since the Schatten $\ell_\infty$-norm of the superoperator 
$\cG \otimes \cH$ is one,
it suffices to
show that for any matrix $\gamma' \in (W')^\perp$ with 
$\elltwo{\gamma'} = 1$,
$
\elltwo{(\cG \zigzag \cH)(\gamma')} \leq 
\lambda - 2 \sqrt[4]{\frac{t(t-1)}{d}}.
$
This is equivalent to showing that
\[
|\inprod{\delta'}{(\cG \zigzag \cH)(\gamma')}| \leq 
\lambda - 2 \sqrt[4]{\frac{t(t-1)}{d}}
\]
for any matrices $\gamma', \delta' \in (W')^\perp$ with 
$\elltwo{\gamma'} = \elltwo{\delta'} = 1$.

Let us write 
$\gamma' = a_1 \gamma'_1 + a_1 \gamma'_2$, 
$\delta' = b_1 \delta'_1 + b_2 \delta'_2$, 
where
\[
\gamma'_1, \delta'_1 
\in (\C^{D^t \times D^t} \otimes (W'_2)) \cap (W')^\perp,
~~~
\gamma'_2, \delta'_2 
\in \C^{D^t \times D^t} \otimes (W'_2)^\perp,
\]
$a_1, b_1, a_2, b_2 \in \C$, 
$
\elltwo{\gamma'_1} = \elltwo{\delta'_1} = 
\elltwo{\gamma'_2} = \elltwo{\delta'_2} =
1.
$
Since $\gamma'_1$ is orthogonal to $\gamma'_2$ and
$\delta'_1$ is orthogonal to $\delta'_2$, 
$|a_1|^2 + |a_2|^2 = \elltwo{\gamma'}^2 = 1$ and 
$|b_1|^2 + |b_2|^2 = \elltwo{\delta'}^2 = 1$.
Thus,
\begin{eqnarray*}
\lefteqn{\inprod{\delta'}{(\cG \zigzag \cH)(\gamma')}}\\
& = &
b_1^* a_1 \inprod{\delta'_1}{(\cG \zigzag \cH)(\gamma'_1)} +
b_1^* a_2 \inprod{\delta'_1}{(\cG \zigzag \cH)(\gamma'_2)}) \\
&   &
{} +
b_2^* a_1 \inprod{\delta'_2}{(\cG \zigzag \cH)(\gamma'_1)}) +
b_2^* a_2 \inprod{\delta'_2}{(\cG \zigzag \cH)(\gamma'_2}.
\end{eqnarray*}
We will now bound each of the four inner products in the above equation.

We bound the fourth inner product first.
By Lemma~\ref{lem:subspaces}, there are matrices
$\gamma_2, \delta_2 \in \C^{D^t \times D^t} \otimes (W_2)^\perp$ such that
\[
\elltwo{\gamma'_2 - \gamma_2}, \elltwo{\delta'_2 - \delta_2} \leq 
2 \sqrt[4]{\frac{t(t-1)}{d}},
~~~
\elltwo{\gamma_2} \leq 1, ~  \elltwo{\delta_2} \leq 1.
\]
Hence,
\begin{eqnarray*}
\elltwo{(\I^{\C^{D^t \times D^t}} \otimes \cH)(\gamma'_2)} 
& \leq &
\elltwo{(\I^{\C^{D^t \times D^t}} \otimes \cH)(\gamma_2)} +
\elltwo{\gamma'_2 - \gamma_2} \\
& \leq &
\lambda_2 \elltwo{\gamma_2} + 2 \sqrt[4]{\frac{t(t-1)}{d}} \leq
\lambda_2 + 2 \sqrt[4]{\frac{t(t-1)}{d}}.
\end{eqnarray*}
Similarly,
\begin{eqnarray*}
\elltwo{(\I^{\C^{D^t \times D^t}} \otimes \cH^\dag)(\delta'_2)} 
& \leq &
\elltwo{(\I^{\C^{D^t \times D^t}} \otimes \cH^\dag)(\delta_2)} +
\elltwo{\delta'_2 - \delta_2} \\
& \leq &
\lambda_2 \elltwo{\delta_2} + 2 \sqrt[4]{\frac{t(t-1)}{d}} \leq
\lambda_2 + 2 \sqrt[4]{\frac{t(t-1)}{d}},
\end{eqnarray*}
where, in the second inequality above, we used the fact that the 
right singular vectors of $\cH^\dag$
are the left singular vectors of $\cH$ and vice versa, and the fact
that $(W_2)^\perp$ is the span of the left as well as the right singular 
vectors of $\cH$ with singular value at most $\lambda_2$.
Thus,
\begin{eqnarray*}
\lefteqn{
|\inprod{\delta'_2}{(\cG \zigzag \cH)(\gamma'_2)}|
} \\
& = &
|\inprod{\delta'_2}{
(
(\I^{\C^{D^t \times D^t}} \otimes \cH) \circ 
\ddcG^{\otimes t} \circ
(\I^{\C^{D^t \times D^t}} \otimes \cH)
)
(\gamma'_2)
}| \\
& = &
|\inprod{(\I^{\C^{D^t \times D^t}} \otimes \cH^\dag)(\delta'_2)}
{
\dcG^{\otimes t} 
((\I^{\C^{D^t \times D^t}} \otimes \cH)(\gamma'_2))
(\dcG^\dagger)^{\otimes t}
}| \\
& \leq &
\elltwo{(\I^{\C^{D^t \times D^t}} \otimes \cH^\dag)(\delta'_2)}
\elltwo{
\dcG^{\otimes t} 
((\I^{\C^{D^t \times D^t}} \otimes \cH)(\gamma'_2))
(\dcG^\dagger)^{\otimes t}
} \\
&   =  &
\elltwo{(\I^{\C^{D^t \times D^t}} \otimes \cH^\dag)(\delta'_2)}
\elltwo{
((\I^{\C^{D^t \times D^t}} \otimes \cH)(\gamma'_2))
} \\
& \leq &
\left(\lambda_2 + 2\sqrt[4]{\frac{t(t-1)}{d}}\right)^2.
\end{eqnarray*}

We bound the second and third inner products similarly.
\begin{eqnarray*}
\lefteqn{
|\inprod{\delta'_1}{(\cG \zigzag \cH)(\gamma'_2)}|
} \\
& = &
|\inprod{(\I^{\C^{D^t \times D^t}} \otimes \cH^\dag)(\delta'_1)}
{
\dcG^{\otimes t} 
((\I^{\C^{D^t \times D^t}} \otimes \cH)(\gamma'_2))
(\dcG^\dagger)^{\otimes t}
}| \\
& \leq &
\elltwo{(\I^{\C^{D^t \times D^t}} \otimes \cH^\dag)(\delta'_1)}
\elltwo{
((\I^{\C^{D^t \times D^t}} \otimes \cH)(\gamma'_2))
} 
\;\leq\;
\lambda_2 + 2\sqrt[4]{\frac{t(t-1)}{d}},
\end{eqnarray*}
where we used the fact that the Schatten $\ell_\infty$-norm of
the superoperator $\I^{\C^{D^t \times D^t}} \otimes \cH^\dag$
is one in the inequality above. Analogously,
\begin{eqnarray*}
\lefteqn{
|\inprod{\delta'_2}{(\cG \zigzag \cH)(\gamma'_1)}|
} \\
& = &
|\inprod{(\I^{\C^{D^t \times D^t}} \otimes \cH^\dag)(\delta'_2)}
{
\dcG^{\otimes t} 
((\I^{\C^{D^t \times D^t}} \otimes \cH)(\gamma'_1))
(\dcG^\dagger)^{\otimes t}
}| \\
& \leq &
\elltwo{(\I^{\C^{D^t \times D^t}} \otimes \cH^\dag)(\delta'_2)}
\elltwo{
((\I^{\C^{D^t \times D^t}} \otimes \cH)(\gamma'_1))
} 
\;\leq\;
\lambda_2 + 2\sqrt[4]{\frac{t(t-1)}{d}}.
\end{eqnarray*}

Finally, we bound the first inner product as follows.
Let 
$
(\C^{D^t \times D^t} \otimes W'_2) \cap (W')^\perp \ni \gamma'_1 =
\sum_{\sigma} b_\sigma (\beta_1)_\sigma \otimes (\alpha'_2)_\sigma,
$
where $\elltwo{\gamma'_1} = 1$, 
$\elltwo{(\beta_1)_\sigma} = 1$, $b_\sigma \in \C$. 
Similarly, let
$
(\C^{D^t \times D^t} \otimes W'_2) \cap (W')^\perp \ni \delta'_1 =
\sum_{\sigma} c_\sigma (\theta_1)_\sigma \otimes (\alpha'_2)_\sigma,
$
where $\elltwo{\delta'_1} = 1$, 
$\elltwo{(\theta_1)_\sigma} = 1$, $c_\sigma \in \C$. 
Let $b$ be the $t!$-tuple whose $\sigma$th entry
is the complex number $b_\sigma$.
Let $c$ be the $t!$-tuple whose $\sigma$th entry
is the complex number $c_\sigma$.
By Lemma~\ref{lem:perm},
\[
1 = 
\elltwo{\gamma'_1}^2 = 
\sum_\sigma |b_\sigma|^2 \elltwo{(\alpha'_2)_\alpha}^2 \geq
\elltwo{b}^2 \left(1 - \frac{t(t-1)}{2d}\right),
\]
which gives
$
\elltwo{b}^2 \leq \frac{1}{1 - \frac{t(t-1)}{2d}}.
$
Similarly,
$
\elltwo{c}^2 \leq \frac{1}{1 - \frac{t(t-1)}{2d}}.
$
Define the matrices
\[
\gamma_1 :=
\sum_{\sigma} b_\sigma (\beta_1)_\sigma \otimes (\alpha_2)_\sigma,
~~~
\delta_1 :=
\sum_{\sigma} c_\sigma (\theta_1)_\sigma \otimes (\alpha_2)_\sigma
~~~
\gamma_1, \delta_1 \in \C^{D^t \times D^t} \otimes W_2.
\]
By Lemma~\ref{lem:subspaces},
$
\elltwo{\gamma'_1 - \gamma_1},
\elltwo{\delta'_1 - \delta_1} \leq 2\sqrt{\frac{t(t-1)}{d}}.
$
Recalling the fact that the Schatten $\ell_\infty$-norm of 
$\cG \zigzag \cH$ is one, we get
\begin{eqnarray*}
\lefteqn{
|
\inprod{\delta'_1}{(\cG \zigzag \cH)(\gamma'_1)} -
\inprod{\delta_1}{(\cG \zigzag \cH)(\gamma_1)}
| 
} \\
& \leq &
|
\inprod{\delta'_1}{(\cG \zigzag \cH)(\gamma'_1)} -
\inprod{\delta'_1}{(\cG \zigzag \cH)(\gamma_1)}
| +
|
\inprod{\delta'_1}{(\cG \zigzag \cH)(\gamma_1)} -
\inprod{\delta_1}{(\cG \zigzag \cH)(\gamma_1)}
| \\
& \leq &
2\sqrt{\frac{t(t-1)}{d}} +
2\sqrt{\frac{t(t-1)}{d}} 
\left(1 + 2\sqrt{\frac{t(t-1)}{d}}\right) 
\;\leq\;
9\sqrt{\frac{t(t-1)}{d}}.
\end{eqnarray*}
Moreover, as $\gamma'_1 \in (W')^\perp$, 
$\inprod{(\beta_1)_\sigma}{(\alpha_1)_\sigma} = 0$ 
for all $\sigma \in S_t$.
We now evaluate
\begin{eqnarray*}
\lefteqn{
\inprod{\delta_1}{(\cG \zigzag \cH)(\gamma_1)}
} \\
& = &
\inprod{\delta_1}{
(
(\I^{\C^{D^t \times D^t}} \otimes \cH) \circ 
\ddcG^{\otimes t} \circ
(\I^{\C^{D^t \times D^t}} \otimes \cH)
)
(\gamma_1)
} \\
& = &
\inprod{(\I^{\C^{D^t \times D^t}} \otimes \cH^\dag)(\delta_1)}{
\dcG^{\otimes t} 
((\I^{\C^{D^t \times D^t}} \otimes \cH)(\gamma_1))
(\dcG^\dagger)^{\otimes t}
} \\
& = &
\inprod{\delta_1}{\dcG^{\otimes t} \gamma_1 (\dcG^\dagger)^{\otimes t}} \\
& = &
\sum_{\sigma', \sigma} c_{\sigma'}^* b_\sigma
\inprod{(\theta_1)_{\sigma'} \otimes (\alpha_2)_{\sigma'}}{
\dcG^{\otimes t} 
((\beta_1)_\sigma \otimes (\alpha_2)_\sigma)
(\dcG^\dagger)^{\otimes t}
} \\
& = &
\sum_{\sigma' \neq \sigma} c_{\sigma'}^* b_\sigma
\inprod{(\theta_1)_{\sigma'} \otimes (\alpha_2)_{\sigma'}}{
\dcG^{\otimes t} 
((\beta_1)_\sigma \otimes (\alpha_2)_\sigma)
(\dcG^\dagger)^{\otimes t}
} \\
&    &
{} +
\sum_{\sigma} c_\sigma^* b_\sigma
\inprod{(\theta_1)_{\sigma} \otimes (\alpha_2)_{\sigma}}{
\dcG^{\otimes t} 
((\beta_1)_\sigma \otimes (\alpha_2)_\sigma)
(\dcG^\dagger)^{\otimes t}
}.
\end{eqnarray*}
Fix $\sigma, \sigma' \in S_t$, $\sigma \neq \sigma'$. Then,
\begin{eqnarray*}
\lefteqn{
|
\inprod{(\theta_1)_{\sigma'} \otimes (\alpha_2)_{\sigma'}}{
\dcG^{\otimes t} 
((\beta_1)_\sigma \otimes (\alpha_2)_\sigma)
(\dcG^\dagger)^{\otimes t}
}
|
} \\
& = &
d^{-t}
\left|
\sum_{i_1, \ldots, i_t, j_1, \ldots, j_t \in [d]}
\langle
(\theta_1)_{\sigma'} \otimes 
(
(e_{i_1} \otimes \cdots \otimes e_{i_t}) 
(e_{i_{\sigma'(1)}}^\dag \otimes \cdots \otimes e_{i_{\sigma'(t)}}^\dag)
), 
\right. \\
&   &
~~~~~~~~~~~~~~~~~~~~~~~~~~~~~~
(
(U_{j_1} \otimes \cdots \otimes U_{j_t})(\beta_1)_{\sigma} 
(
U_{j_{\sigma(1)}}^\dagger \otimes \cdots \otimes 
U_{j_{\sigma(t)}}^\dagger
)
) \\
&   &
\left.
~~~~~~~~~~~~~~~~~~~~~~~~~~~~~~~~~~~~~~~
{} \otimes
(
(e_{-j_1} \otimes \cdots \otimes e_{-j_t})
(e_{-j_{\sigma(1)}}^\dag \otimes \cdots \otimes e_{-j_{\sigma(t)}}^\dag)
)
\rangle
\right| \\
& \leq &
d^{-t}
\sum_{i_1, \ldots, i_t, j_1, \ldots, j_t \in [d]}
|
\inprod{(\theta_1)_{\sigma'}}
{
(U_{j_1} \otimes \cdots \otimes U_{j_t})(\beta_1)_{\sigma} 
(
U_{j_{\sigma(1)}}^\dagger \otimes \cdots \otimes 
U_{j_{\sigma(t)}}^\dagger
)
}
| \\
&      &
~~~~~~~~~~~~~~~~~~~~~~~~~~~~~~~
{} \cdot |
\langle
(e_{i_1} \otimes \cdots \otimes e_{i_t})
(e_{i_{\sigma'(1)}}^\dag \otimes \cdots \otimes e_{i_{\sigma'(t)}}^\dag), 
\\
&      & 
~~~~~~~~~~~~~~~~~~~~~~~~~~~~~~~~~~~~~~~~~~~
(e_{-j_1} \otimes \cdots \otimes e_{-j_t})
(e_{-j_{\sigma(1)}}^\dag \otimes \cdots \otimes e_{-j_{\sigma(t)}}^\dag)
\rangle| \\
& \leq &
d^{-t}
\sum_{i_1, \ldots, i_t}
\delta_{i_{\sigma'(1)}, i_{\sigma(1)}} \cdots 
\delta_{i_{\sigma'(t)}, i_{\sigma(t)}} 
\;  = \;
d^{t_{(\sigma')^\dagger \sigma} - t}.
\end{eqnarray*}
Hence,
\begin{eqnarray*}
\lefteqn{
\left|\sum_{\sigma' \neq \sigma} c_{\sigma'}^* b_\sigma
\inprod{(\theta_1)_{\sigma'} \otimes (\alpha_2)_{\sigma'}}{
\dcG^{\otimes t} 
((\beta_1)_\sigma \otimes (\alpha_2)_\sigma)
(\dcG^\dagger)^{\otimes t}
}\right|
} \\
& \leq &
\sum_{\sigma' \neq \sigma} |c_{\sigma'}| |b_\sigma|
|
\inprod{(\theta_1)_{\sigma'} \otimes (\alpha_2)_{\sigma'}}{
\dcG^{\otimes t} 
((\beta_1)_\sigma \otimes (\alpha_2)_\sigma)
(\dcG^\dagger)^{\otimes t}
}
| \\
& \leq &
\sum_{\sigma' \neq \sigma} |c_{\sigma'}| |b_\sigma|
d^{t_{(\sigma')^\dagger \sigma} - t} 
\;  = \;
|c|^\dagger M |b|
\;\leq\;
\elltwo{c} \elltwo{b} \ellinfty{M} \\
& \leq &
\frac{t(t-1)}{d}  \frac{1}{1 - \frac{t(t-1)}{2d}} 
\;\leq\;
\frac{2t(t-1)}{d},
\end{eqnarray*}
where $|c|$, $|b|$ denote the $t!$-tuples
whose $\sigma$th entries are $|c_\sigma|$, $|b_\sigma|$, and
Lemma~\ref{lem:stirling} is used in the next to last inequality.
Now fix $\sigma \in S_t$. We have,
\begin{eqnarray*}
\lefteqn{
|
\inprod{(\theta_1)_{\sigma} \otimes (\alpha_2)_{\sigma}}{
\dcG^{\otimes t} 
((\beta_1)_\sigma \otimes (\alpha_2)_\sigma)
(\dcG^\dagger)^{\otimes t}
}
|
} \\
& = &
d^{-t}
\left|
\sum_{i_1, \ldots, i_t, j_1, \ldots, j_t \in [d]}
\langle
(\theta_1)_{\sigma} \otimes 
(
(e_{i_1} \otimes \cdots \otimes e_{i_t})
(e_{i_{\sigma(1)}}^\dag \otimes \cdots \otimes e_{i_{\sigma(t)}}^\dag)
), 
\right. \\
&   &
~~~~~~~~~~~~~~~~~~~~~~~~~~~~~~
(
(U_{j_1} \otimes \cdots \otimes U_{j_t})(\beta_1)_{\sigma} 
(
U_{j_{\sigma(1)}}^\dagger \otimes \cdots \otimes U_{j_{\sigma(t)}}^\dagger
)
) \\
&   &
\left.
~~~~~~~~~~~~~~~~~~~~~~~~~~~~~~~~~~~~~~~
{} \otimes
(
(e_{-j_1} \otimes \cdots \otimes e_{-j_t})
(e_{-j_{\sigma(1)}}^\dag \otimes \cdots \otimes e_{-j_{\sigma(t)}}^\dag)
)
\rangle
\right| \\
& = &
d^{-t}
\left|
\sum_{i_1, \ldots, i_t \in [d]}
\inprod{(\theta_1)_{\sigma}}{
(U_{-i_1} \otimes \cdots \otimes U_{-i_t})(\beta_1)_{\sigma} 
(
U_{-i_{\sigma(1)}}^\dagger \otimes \cdots \otimes 
U_{-i_{\sigma(t)}}^\dagger
)
}
\right| \\
& = &
\left|
\left\langle
(\theta_1)_{\sigma},
d^{-t}
\sum_{i_1, \ldots, i_t \in [d]}
(U_{-i_1} \otimes \cdots \otimes U_{-i_t})(\beta_1)_{\sigma} 
(
U_{-i_{\sigma(1)}}^\dagger \otimes \cdots \otimes 
U_{-i_{\sigma(t)}}^\dagger
)
\right\rangle
\right| \\
& \leq &
\elltwo{(\theta_1)_{\sigma}} \cdot
\elltwo{
d^{-t}
\sum_{i_1, \ldots, i_t \in [d]}
(U_{-i_1} \otimes \cdots \otimes U_{-i_t})(\beta_1)_{\sigma} 
(
U_{-i_{\sigma(1)}}^\dagger \otimes \cdots \otimes 
U_{-i_{\sigma(t)}}^\dagger
)
} \\
&   =  &
\elltwo{
d^{-t}
\sum_{i_1, \ldots, i_t \in [d]}
(U_{i_1} \otimes \cdots \otimes U_{i_t})(\beta_1)_{\sigma} 
(
U_{i_{\sigma(1)}}^\dagger \otimes \cdots \otimes U_{i_{\sigma(t)}}^\dagger
)
}.
\end{eqnarray*}
Let us now express $(\beta_1)_\sigma$ as 
$\Sigma^{(\C^D)^{\otimes t}} (\hat{\beta}_1)_\sigma$
for some $(\hat{\beta}_1)_\sigma \in \C^{D^t \times D^t}$,
$\elltwo{(\hat{\beta}_1)_\sigma} = 1$. Observe that
\[
0 =
\inprod{(\beta_1)_\sigma}{(\alpha_1)_\sigma} =
\inprod{\Sigma^{(\C^D)^{\otimes t}} (\hat{\beta}_1)_\sigma}
       {\Sigma^{(\C^D)^{\otimes t}} \alpha_1} =
\inprod{(\hat{\beta}_1)_\sigma}{\alpha_1}.
\]
We can now write
\begin{eqnarray*}
\lefteqn{
|
\inprod{(\theta_1)_{\sigma} \otimes (\alpha_2)_{\sigma}}{
\dcG^{\otimes t} 
((\beta_1)_\sigma \otimes (\alpha_2)_\sigma)
(\dcG^\dagger)^{\otimes t}
}
|
} \\
& \leq &
\elltwo{
d^{-t}
\sum_{i_1, \ldots, i_t \in [d]}
(U_{i_1} \otimes \cdots \otimes U_{i_t}) 
\Sigma^{(\C^D)^{\otimes t}} (\hat{\beta}_1)_{\sigma} 
(
U_{i_{\sigma(1)}}^\dagger \otimes \cdots \otimes U_{i_{\sigma(t)}}^\dagger
)
} \\
&  =   &
\elltwo{
d^{-t}
\sum_{i_1, \ldots, i_t \in [d]}
\Sigma^{(\C^D)^{\otimes t}}
(U_{i_1} \otimes \cdots \otimes U_{i_t}) (\hat{\beta}_1)_{\sigma} 
(
U_{i_1}^\dagger \otimes \cdots \otimes U_{i_t}^\dagger
)
} \\
&  =   &
\elltwo{
\Sigma^{(\C^D)^{\otimes t}} \cG^{\otimes t}((\hat{\beta}_1)_{\sigma})
} 
\; = \;
\elltwo{
\cG^{\otimes t}((\hat{\beta}_1)_{\sigma})
}.
\end{eqnarray*}
Since $\cG$ is a quantum expander, i.e. a $(D, d, \lambda_1, 1)$-qTPE,
$\cG^{\otimes t}$ is also a quantum expander, i.e. a
$(D^t, d^t, \lambda_1, 1)$-qTPE by Fact~\ref{fact:tensor}. Since 
$\inprod{(\hat{\beta}_1)_\sigma}{\alpha_1} = 0$, we see that
$
\elltwo{
\cG^{\otimes t}((\hat{\beta}_1)_{\sigma})
} \leq \lambda_1.
$
Thus,
\[
|
\inprod{(\theta_1)_{\sigma} \otimes (\alpha_2)_{\sigma}}{
\dcG 
((\beta_1)_\sigma \otimes (\alpha_2)_\sigma)
\dcG^\dagger
}
| \leq 
\lambda_1.
\]
Hence,
\begin{eqnarray*}
\lefteqn{
\left|\sum_{\sigma} c_\sigma^* b_\sigma
\inprod{(\theta_1)_{\sigma} \otimes (\alpha_2)_{\sigma}}{
\dcG^{\otimes t} 
((\beta_1)_\sigma \otimes (\alpha_2)_\sigma)
(\dcG^\dagger)^{\otimes t}
}\right|
} \\
& \leq &
\sum_{\sigma} |c_\sigma| |b_\sigma|
|
\inprod{(\theta_1)_{\sigma} \otimes (\alpha_2)_{\sigma}}{
\dcG^{\otimes t} 
((\beta_1)_\sigma \otimes (\alpha_2)_\sigma)
(\dcG^\dagger)^{\otimes t}
}
| 
\;\leq\;
\lambda_1 \elltwo{c} \elltwo{b} \\
& \leq &
\frac{\lambda_1}{1 - \frac{t(t-1)}{2d}} 
\;\leq\;
\lambda_1\left(1 + \frac{t(t-1)}{d}\right).
\end{eqnarray*}
This implies that
\[
|\inprod{\delta_1}{(\cG \zigzag \cH)(\gamma_1)}| \leq
\lambda_1\left(1 + \frac{t(t-1)}{d}\right) +
2 \frac{t(t-1)}{d},
\]
which further leads to
\begin{eqnarray*}
|\inprod{\delta'_1}{(\cG \zigzag \cH)(\gamma'_1)}| 
& \leq &
\lambda_1\left(1 + \frac{t(t-1)}{d}\right) +
2 \frac{t(t-1)}{d} + 9 \sqrt{\frac{t(t-1)}{d}} \\
& \leq &
\lambda_1 + 12 \sqrt{\frac{t(t-1)}{d}}.
\end{eqnarray*}

Putting the bounds on the four inner products together, we get
\begin{eqnarray*}
\lefteqn{|\inprod{\delta'}{(\cG \zigzag \cH)(\gamma')}|}\\
& \leq &
|b_1| |a_1| \left(\lambda_1 + 12\sqrt{\frac{t(t-1)}{d}}\right) +
|b_1| |a_2| \left(\lambda_2 + 2\sqrt[4]{\frac{t(t-1)}{d}}\right) \\
&     &
{} +
|b_2| |a_1| \left(\lambda_2 + 2\sqrt[4]{\frac{t(t-1)}{d}}\right) +
|b_2| |a_2| \left(\lambda_2 + 2\sqrt[4]{\frac{t(t-1)}{d}}\right)^2 \\
& \leq &
\left(\lambda_1 + 12\sqrt{\frac{t(t-1)}{d}}\right) +
\left(\lambda_2 + 2\sqrt[4]{\frac{t(t-1)}{d}}\right) +
\left(\lambda_2 + 2\sqrt[4]{\frac{t(t-1)}{d}}\right)^2,
\end{eqnarray*}
where we used 
$
|b_1| |a_2| + |b_2| |a_1| \leq
\sqrt{|b_1|^2 + |b_2|^2} \sqrt{|a_2|^2 + |a_1|^2} \leq
1
$
in the last inequality.
This leads to the bound
\begin{eqnarray*}
\lefteqn{\elltwo{(\cG \zigzag \cH)(\gamma)}} \\
& \leq &
\left(\lambda_1 + 12\sqrt{\frac{t(t-1)}{d}}\right) +
\left(\lambda_2 + 2\sqrt[4]{\frac{t(t-1)}{d}}\right) +
\left(\lambda_2 + 2\sqrt[4]{\frac{t(t-1)}{d}}\right)^2 +
2\sqrt[4]{\frac{t(t-1)}{d}} \\
& \leq &
\lambda_1 + \lambda_2 + \lambda_2^2 + 
24\sqrt[4]{\frac{t(t-1)}{d}}, 
\end{eqnarray*}
finishing the proof of the theorem.
\end{proof}

\noindent
\paragraph{Remarks:} \  

\smallskip

\noindent
1.\ Setting $t=1$ recovers the eigenvalue bound on the zigzag product
of quantum expanders, i.e. $1$-qTPEs, proved in
\cite[Theorem~4.8]{BST}.

\smallskip

\noindent
2.\ For $t = \polylog(D)$, taking an efficient construction 
(e.g. via the zigzag product) of a 
$(D, d, \lambda_1, 1)$-qTPE, $d = (10 s t \log t)^6$, 
$\lambda_1 = 100 d^{-1/4}$  as in \cite{BST}, and combining it via
the zigzag product with
a $(d, s, \lambda_2, t)$-qTPE, 
$\lambda_2 = 8 s^{-1/2}$ obtained from the random 
construction of
Fact~\ref{fact:randomqTPE}, gives us a $(Dd, s^2, \lambda, t)$-qTPE,
$\lambda := \lambda_1 + 2 \lambda_2 + O(\sqrt[4]{\frac{t^2}{d}})$
which is efficiently computable.
This gives rise to a fourth power tradeoff between
degree $s^2$ and second largest singular value $10 s^{-1/2}$. 
This tradeoff is the same
as in the standard zigzag product for classical \cite{RVW} and quantum 
\cite{BST} expanders.

\smallskip

\noindent
3.\ The reader may wonder why we went from the subspace $W^\perp$ to
$(W')^\perp$ and back in the above proof. The reason behind this
seemingly unnatural strategy is because we want to ensure that
in the proof of the bound on the first inner product, $(\beta_1)_\sigma$ is
perfectly orthogonal to $(\alpha_1)_\sigma$. Approximate orthogonality
in this step seems to give additive losses of $\poly(\frac{t!}{d})$ in the
expression for $\lambda$, which would require $d \geq t!$, leading
to the construction of efficient $t$-qTPEs in dimension $N$ only for 
$t \leq \frac{\log \log N}{\log \log \log N}$.  This is too small for 
many applications. Going from $W^\perp$ to $(W')^\perp$ allows us to
use the fact that $(\alpha'_2)_{\sigma'}$ is orthogonal to
$(\alpha_2)_\sigma$, $\sigma' \neq \sigma$ which
finally ensures that $(\beta_1)_\sigma$ is indeed
perfectly orthogonal to $(\alpha_1)_\sigma$. But then 
the second eigenvalue bounds on $\cG$ and $\cH$ are in terms of $W^\perp$
and so we have to go back to $W^\perp$ from $(W')^\perp$ in order to
use them in the proof. By adopting this back and forth strategy, we 
only get additive
losses of $\poly(\frac{t}{d})$, which would require $d \geq \poly(t)$, 
leading
to the construction of efficient $t$-qTPEs in dimension $N$ for 
$t = \polylog(N)$. 

\smallskip

\noindent
4.\  An improved analysis of $\lambda$ in the above theorem along the
lines of \cite[Theorem~4.3]{RVW} can be done, giving us the bound 
\[
\lambda :=
\frac{1}{2} (1-\mu_2^2) \mu_1 +
\frac{1}{2} \sqrt{(1- \mu_2^2) \mu_1^2 + 4 \mu_2^2} + 
2\sqrt[4]{\frac{t(t-1)}{d}},
\]
where
\[
\mu_1 := \lambda_1 + 9\sqrt{\frac{t(t-1)}{d}},
~~~
\mu_2 := \lambda_2 + 2\sqrt[4]{\frac{t(t-1)}{d}}.
\]
This bound has several nice properties e.g. it is always less than 
$\mu_1 + \mu_2 + 2\sqrt[4]{\frac{t(t-1)}{d}}$, it is always
less than $1 + 2\sqrt[4]{\frac{t(t-1)}{d}}$ if $\mu_1, \mu_2 < 1$ etc.

\smallskip

\noindent
5. As in \cite[Theorem~6.2]{RVW}, one can similarly define a
`derandomised' zigzag product as follows:
\begin{definition}[Derandomised zigzag product of qTPEs]
\label{def:qzigzagprime}
The {\em derandomised zigzag product} of explicitly Hermitian 
qTPEs $\cG$ and $\cH$, denoted by
$\cG \zigzag' \cH$, is defined as the following set
of $s^3$ unitary matrices on $\C^{Dd}$:
\[
\cG \zigzag' \cH :=
\{
(\one^{\C^D} \otimes V_i) 
(\one^{\C^D} \otimes V_j^\dagger)
\dcG
(\one^{\C^D} \otimes V_j)
(\one^{\C^D} \otimes V_k):
i, j, k \in [s]
\}.
\]
\end{definition}
\noindent
With this definition, one can similarly show that the second eigenvalue
$\lambda$ of $\cG \zigzag' \cH$ satisfies the bound
\[
\lambda :=
\mu_1 + 2\mu_2^2 +
2\sqrt[4]{\frac{t(t-1)}{d}},
\]
where
$\mu_1$, $\mu_2$ are defined in the previous remark.
For $t = \polylog(D)$,
using the derandomised zigzag product for constructing a
quantum expander i.e. $(D, d, \lambda_1, 1)$-qTPE, $d = (10 s t \log t)^6$,
$\lambda_1 = 100 d^{-1/3}$, and combining it
via the derandomised zigzag product
with a $(d, s, \lambda_2, t)$-qTPE, $\lambda_2 = 8 s^{-1/2}$ obtained 
from Fact~\ref{fact:randomqTPE}, gives us a
$(Dd, s^3, \lambda, t)$-qTPE,
$\lambda := \lambda_1 + 2 \lambda_2^2 + O(\sqrt[4]{\frac{t^2}{d}})$
which is efficiently computable.
This gives rise to a third power tradeoff between
degree $s^3$ and second largest singular value $130 s^{-1}$. 
This tradeoff is the same
as in the derandomised zigzag product for classical expanders \cite{RVW}.

\section{Generalised zigzag product gives almost Ramanujan qTPE}
Inspired by the definition of generalised zigzag product for 
classical expanders, i.e.  $1$-cTPEs, in \cite{BT}, 
we define the zigzag product of a $1$-qTPE and a $t$-qTPE as follows.
\begin{definition}[Generalised zigzag product of qTPEs]
\label{def:qgenzigzag}
Let $\cG = \{U_i\}_{i=1}^d$ be a $(D, d, \lambda_1, 1)$-qTPE. For 
$1 \leq j \leq k$, let $\cH_j = \{V_i(j)\}_{i=1}^s$ be a
$(dd', s, \lambda_2, t)$-qTPE. Let $\vcH := (\cH_k, \ldots, \cH_1)$.
Define the unitary matrix $\dcG$ on the vector space
$\C^{Ddd'} \cong \C^D \otimes (\C^{d} \otimes \C^{d'})$ by
\[
e_a \otimes (e_b \otimes e_{b'}) 
\stackrel{\dcG}{\mapsto}
(U_b e_a) \otimes (e_b \otimes e_{b'}),
\]
where $e_a$, $e_b$, $e_{b'}$ denote computational basis vectors of
$\C^D$, $\C^d$, $\C^{d'}$ respectively.
The zigzag product of qTPEs $\cG$ and $\vcH$, denoted by
$\cG \zigzag \vcH$, is defined as the following set
of $s^{k}$ unitary matrices on $\C^{Ddd'}$:
\[
\cG \zigzag \cH :=
\{
(\one^{\C^D} \otimes V_{i_{k}}(k)) 
\dcG \cdots \dcG
(\one^{\C^D} \otimes V_{i_1}(1)):
i_k, \ldots, i_1 \in [s]
\}.
\]
\end{definition}

\paragraph{Remarks:} \ \\

\noindent
1.\ The generalised zigzag product of Hermitian qTPEs will in general
not be Hermitian because the qTPEs $\cH_k, \ldots, \cH_1$ in general
have no relation amongst them. That is why we dispense with the
involution `-' in defining the unitary $\dcG$ and the generalised
zigzag product. Note that any qTPE can be made explicitly Hermitian
by doubling its degree according to Fact~\ref{fact:doubling}.

\smallskip

\noindent
2.\ Viewed as a superoperator on $\C^{(Ddd')^t \times (Ddd')^t}$,
the generalised zigzag product $\cG \zigzag \vcH$ is nothing but
\[
\cG \zigzag \vcH :=
(\I^{\C^{D^t \times D^t}} \otimes \cH_k) \circ 
\ddcG^{\otimes t} \circ \cdots \circ \ddcG^{\otimes t}
(\I^{\C^{D^t \times D^t}} \otimes \cH_1).
\]

\bigskip

Suppose $t \leq dd' \leq D \leq Ddd'$. Define the subspaces
\[
\begin{array}{c}
W \leq \C^{(Ddd')^t \times (Ddd')^t}, 
W_1 \leq \C^{D^t \times D^t}, 
W_2 \leq \C^{(dd')^t \times (dd')^t}, \\
W'_2 \leq \C^{(dd')^t \times (dd')^t}, 
W' \leq \C^{(Ddd')^t \times (Ddd')^t}, 
W^\perp \leq \C^{(Ddd')^t \times (Ddd')^t}, \\
(W')^\perp \leq \C^{(Ddd')^t \times (Ddd')^t}, 
(W_1)^\perp \leq \C^{D^t \times D^t},  
(W'_2)^\perp \leq \C^{(dd')^t \times (dd')^t}, 
\end{array}
\]
and matrices 
\[
\begin{array}{c}
\alpha_\sigma \in \C^{(Ddd')^t \times (Ddd')^t}, 
(\alpha_1)_\sigma \in \C^{D^t \times D^t},  
(\alpha_2)_\sigma \in \C^{(dd')^t \times (dd')^t},  \\
(\alpha'_2)_\sigma \in \C^{(dd')^t \times (dd')^t}, 
\alpha'_\sigma \in \C^{(Ddd')^t \times (Ddd')^t}, 
\end{array}
\]
for a $\sigma \in S_t$ in similar fashion as before.
Then 
\[
(W')^\perp = 
(\C^{D^t \times D^t} \otimes (W'_2)^\perp) \oplus
((\C^{D^t \times D^t} \otimes W'_2) \cap (W')^\perp)
\]
as before and Lemma~\ref{lem:subspaces} on everywhere closeness 
holds with $d$ replaced by $dd'$. 
For a matrix $\gamma \in W^\perp$, define $\gamma'$ to be the matrix
in $(W')^\perp$ such that
\[
\elltwo{\gamma' - \gamma} \leq
2\sqrt[4]{\frac{t(t-1)}{dd'}}
\elltwo{\gamma},
~~~~
\elltwo{\gamma'} \leq \elltwo{\gamma},
\]
whose existence is guaranteed by Lemma~\ref{lem:subspaces}.
For a matrix 
$\gamma' \in (W')^\perp$, define 
$(\gamma')^{\|'}$ to be its projection onto 
$(\C^{D^t \times D^t} \otimes (W'_2)) \cap (W')^\perp$ and
$(\gamma')^{\perp'}$ to be its projection onto
$\C^{D^t \times D^t} \otimes (W'_2)^\perp$.
Define $(\gamma')^{\|}$ to be the matrix in
$\C^{D^t \times D^t} \otimes W_2$ such that
\[
\elltwo{(\gamma')^\| - (\gamma')^{\|'}} \leq
2\sqrt{\frac{t(t-1)}{dd'}}
\elltwo{(\gamma')^{\|'}},
\]
whose existence is guaranteed by Lemma~\ref{lem:subspaces}.
Define $(\gamma')^{\perp}$ to be the matrix in
$\C^{D^t \times D^t} \otimes (W_2)^\perp$ such that
\[
\elltwo{(\gamma')^\perp - (\gamma')^{\perp'}} \leq
2\sqrt[4]{\frac{t(t-1)}{dd'}}
\elltwo{(\gamma')^{\perp'}},
~~~~
\elltwo{(\gamma')^\perp} \leq \elltwo{(\gamma')^{\perp'}},
\]
whose existence is guaranteed by Lemma~\ref{lem:subspaces}.
We now define by induction two sequences of matrices starting with
$\gamma_0, \delta_0 \in W^\perp$, 
$\elltwo{\gamma_0} = \elltwo{\delta_0} = 1$. 
Let $\gamma'_0$, $\delta'_0$  be the matrices in $(W')^\perp$ that are
$2\sqrt[4]{\frac{t(t-1)}{dd'}}$-close to $\gamma_0$, $\delta_0$, whose
existence has been shown above.
For $1 \leq i \leq k-1$, define 
\[
\gamma_i := 
(\ddcG^{\otimes t} \circ (\I \otimes \cH_i))(\gamma'_{i-1})^{\perp},
~~~
\delta_i := 
((\ddcG^\dagger)^{\otimes t} \circ (\I \otimes \cH_{k-i+1})^\dagger)
(\delta'_{i-1})^{\perp}.
\]
Observe that $\gamma_i, \delta_i \in W^\perp$, so we can define
$\gamma'_i, \delta'_i \in (W')^\perp$ accordingly. This implies
that 
$
(\gamma'_i)^{\|'}, (\delta'_i)^{\|'} \in 
(\C^{D^t \times D^t} \cap W'_2) \cap (W')^\perp.
$
We will assume that our parameters are such that $\lambda_2 < 1/2$.
Using induction on $i$, it is easy to see that
\[
\elltwo{\gamma'_i} \leq 
\elltwo{\gamma_i} \leq 
\lambda_2 \elltwo{(\gamma'_{i-1})^\perp} \leq 
\lambda_2 \elltwo{(\gamma'_{i-1})^{\perp'}} \leq 
\lambda_2 \elltwo{\gamma'_{i-1}} \leq 
\lambda_2^i.
\]
Similarly,
\[
\elltwo{\delta'_i} \leq 
\elltwo{\delta_i} \leq 
\lambda_2 \elltwo{(\delta'_{i-1})^\perp} \leq 
\lambda_2 \elltwo{(\delta'_{i-1})^{\perp'}} \leq 
\lambda_2 \elltwo{\delta'_{i-1}} \leq 
\lambda_2^i.
\]

By construction, we have the orthogonal decompositions
$\gamma'_i = (\gamma'_i)^{\|'} + (\gamma'_i)^{\perp'}$,
$\delta'_i = (\delta'_i)^{\|'} + (\delta'_i)^{\perp'}$. Thus,
$
\elltwo{\gamma'_i}^2 =
\elltwo{(\gamma'_i)^{\|'}}^2 + 
\elltwo{(\gamma'_i)^{\perp'}}^2,
$
$
\elltwo{\delta'_i}^2 =
\elltwo{(\delta'_i)^{\|'}}^2 + 
\elltwo{(\delta'_i)^{\perp'}}^2.
$
By induction, we observe that
\begin{eqnarray*}
\lefteqn{
\sum_{0 \leq i \leq k-1}
\elltwo{(\gamma'_{i})^{\|'}}^2 
} \\
& \leq &
\elltwo{(\gamma'_{0})^{\|'}}^2 + 
\elltwo{\gamma'_{1}}^2 
\;\leq\;
\elltwo{(\gamma'_{0})^{\|'}}^2 + 
\lambda_2^2 \elltwo{(\gamma'_{0})^{\perp}}^2 
\;\leq\;
\elltwo{(\gamma'_{0})^{\|'}}^2 + 
\elltwo{(\gamma'_{0})^{\perp'}}^2 \\
& = &
\elltwo{\gamma'_{0}}^2 
\;\leq\;
\elltwo{\gamma_{0}}^2 
\;=\;
1.
\end{eqnarray*}
Similarly,
\[
\sum_{0 \leq i \leq k-1}
\elltwo{(\delta'_{i})^{\|'}}^2 \leq 1.
\]
On the other hand,
\[
\sum_{0 \leq i \leq k-1}
\elltwo{(\delta'_{i})^{\perp}}^2 \leq
\sum_{0 \leq i \leq k-1}
\elltwo{(\delta'_{i})^{\perp'}}^2 \leq
\sum_{0 \leq i \leq k-1}
\elltwo{\delta'_{i}}^2 \leq
\sum_{0 \leq i \leq k-1} \lambda_2^{2i} \leq
2.
\]
Similarly,
\[
\sum_{0 \leq i \leq k-1}
\elltwo{(\gamma'_{i})^{\perp'}}^2 \leq
2.
\]

For $0 \leq i < j \leq k$, define 
\begin{eqnarray*}
e_i 
& := &
\inprod{\delta_0}{
(
(\I^{\C^{D^t \times D^t}} \otimes \cH_k) \circ 
\ddcG^{\otimes t} \circ \cdots \circ \ddcG^{\otimes t}
(\I^{\C^{D^t \times D^t}} \otimes \cH_1)
)(\gamma_i - \gamma'_i)
}, \\
f_i 
& := &
\inprod{\delta_0}{
(
(\I^{\C^{D^t \times D^t}} \otimes \cH_k) \circ 
\ddcG^{\otimes t} \circ \cdots \circ \ddcG^{\otimes t}
(\I^{\C^{D^t \times D^t}} \otimes \cH_1)
)((\gamma'_i)^{\perp'} - (\gamma'_i)^\perp)
}, \\
d_i 
& := &
\inprod{(\delta'_{k-i-1})^{\perp}}{
(\I^{\C^{D^t \times D^t}} \otimes \cH_{i+1})
((\gamma'_i)^{\|'} - (\gamma'_i)^{\|})
}, \\
l_{i} 
& := &
\inprod{(\delta'_{k-i-1})^{\|}}{
(\I^{\C^{D^t \times D^t}} \otimes \cH_{i+1})
((\gamma'_i)^{\|'} - (\gamma'_i)^{\|})
}, \\
g_{ji}
& := &
\inprod{\delta_{k-j} - \delta'_{k-j}}{
(
(\I^{\C^{D^t \times D^t}} \otimes \cH_k) \circ 
\ddcG^{\otimes t} \circ \cdots \circ \ddcG^{\otimes t}
(\I^{\C^{D^t \times D^t}} \otimes \cH_{i+1})
)((\gamma'_i)^{\|'})
}, \\
h_{ji} 
& := &
\inprod{(\delta'_{k-j})^{\perp'} - (\delta'_{k-j})^{\perp}}{
(
(\I^{\C^{D^t \times D^t}} \otimes \cH_k) \circ 
\ddcG^{\otimes t} \circ \cdots \circ \ddcG^{\otimes t}
(\I^{\C^{D^t \times D^t}} \otimes \cH_1)
)((\gamma'_i)^{\|'})
}, \\
m_{ji} 
& := &
\inprod{(\delta'_{k-j})^{\|'} - (\delta'_{k-j})^{\|}}{
(
(\I^{\C^{D^t \times D^t}} \otimes \cH_j) \circ 
\ddcG^{\otimes t} \circ \cdots \circ \ddcG^{\otimes t}
(\I^{\C^{D^t \times D^t}} \otimes \cH_{i+1})
)((\gamma'_i)^{\|'})
}.
\end{eqnarray*}
Then,
\begin{eqnarray*}
|e_i| 
& \leq &
\elltwo{\gamma_i - \gamma'_i} \leq
2\sqrt{\frac{t(t-1)}{dd'}}
\elltwo{\gamma_i}, \\
|f_i| 
& \leq &
\elltwo{(\gamma'_i)^{\perp'} - (\gamma'_i)^\perp} \leq
2\sqrt[4]{\frac{t(t-1)}{dd'}}
\elltwo{(\gamma'_i)^{\perp'}} \leq
2\sqrt[4]{\frac{t(t-1)}{dd'}}
\elltwo{\gamma_i}, \\
|d_i| 
& \leq &
\elltwo{(\gamma'_i)^{\|'} - (\gamma'_i)^{\|}} \cdot
\elltwo{(\delta'_{k-i-1})^\perp} \leq
2\sqrt{\frac{t(t-1)}{dd'}}
\elltwo{(\gamma'_i)^{\|'}} \cdot
\elltwo{(\delta'_{k-i-1})^{\perp}} \\
& \leq &
\sqrt[4]{\frac{t(t-1)}{dd'}}
\left(
\elltwo{(\gamma'_i)^{\|'}}^2 + \elltwo{(\delta'_{k-i-1})^{\perp}}^2
\right), \\
|l_{i}| 
& \leq &
\elltwo{(\gamma'_i)^{\|'} - (\gamma'_i)^{\|}} \cdot
\elltwo{(\delta'_{k-i-1})^{\|}} \\
& \leq &
2\sqrt{\frac{t(t-1)}{dd'}}
\left(1 + 2\sqrt{\frac{t(t-1)}{dd'}}\right)
\elltwo{(\gamma'_i)^{\|'}} \cdot
\elltwo{(\delta'_{k-i-1})^{\|'}} \\
& \leq &
2\sqrt{\frac{t(t-1)}{dd'}}
(\elltwo{(\gamma'_i)^{\|'}}^2 + \elltwo{(\delta'_{k-i-1})^{\|'}}^2) \\
|g_{ji}| 
& \leq &
\elltwo{\delta_{k-j} - \delta'_{k-j}} \cdot \elltwo{(\gamma'_i)^{\|'}} 
\leq
2\sqrt{\frac{t(t-1)}{dd'}}
\elltwo{\delta_{k-j}} \cdot \elltwo{\gamma_i}, \\
|h_{ji}| 
& \leq &
\elltwo{(\delta'_{k-j})^{\perp'} - (\delta'_{k-j})^\perp} \cdot
\elltwo{(\gamma'_i)^{\|'}} \leq
2\sqrt[4]{\frac{t(t-1)}{dd'}}
\elltwo{(\delta'_{k-j})^{\perp'}} \cdot \elltwo{\gamma_i} \\
& \leq &
2\sqrt[4]{\frac{t(t-1)}{dd'}}
\elltwo{\delta_{k-j}} \cdot \elltwo{\gamma_i}, \\
|m_{ji}| 
& \leq &
\elltwo{(\delta'_{k-j})^{\|'} - (\delta'_{k-j})^{\|}} \cdot
\elltwo{(\gamma'_{i})^{\|'}} \leq
2\sqrt{\frac{t(t-1)}{dd'}}
\elltwo{(\delta'_{k-j})^{\|'}} \cdot
\elltwo{(\gamma'_{i})^{\|'}} \\
& \leq &
2\sqrt{\frac{t(t-1)}{dd'}}
\elltwo{\gamma_i} \cdot \elltwo{\delta_{k-j}}.
\end{eqnarray*}

We can now write
\begin{eqnarray*}
\lefteqn{
\inprod{\delta_0}{(\cG \zigzag \vcH)(\gamma_0)} 
} \\
& = &
\inprod{\delta_0}{
(
(\I^{\C^{D^t \times D^t}} \otimes \cH_k) \circ 
\ddcG^{\otimes t} \circ \cdots \circ \ddcG^{\otimes t}
(\I^{\C^{D^t \times D^t}} \otimes \cH_1)
)(\gamma_0)
}  \\
& = &
\inprod{\delta_0}{
(
(\I^{\C^{D^t \times D^t}} \otimes \cH_k) \circ 
\ddcG^{\otimes t} \circ \cdots \circ \ddcG^{\otimes t}
(\I^{\C^{D^t \times D^t}} \otimes \cH_1)
)(\gamma'_0)
} + e_0 \\
& =  &
\inprod{\delta_0}{
(
(\I^{\C^{D^t \times D^t}} \otimes \cH_k) \circ 
\ddcG^{\otimes t} \circ \cdots \circ \ddcG^{\otimes t}
(\I^{\C^{D^t \times D^t}} \otimes \cH_1)
)((\gamma'_0)^{\perp'})
} \\
&    &
{} +
\inprod{\delta_0}{
(
(\I^{\C^{D^t \times D^t}} \otimes \cH_k) \circ 
\ddcG^{\otimes t} \circ \cdots \circ \ddcG^{\otimes t}
(\I^{\C^{D^t \times D^t}} \otimes \cH_1)
)((\gamma'_0)^{\|'})
} + 
e_0 \\
& =  &
\inprod{\delta_0}{
(
(\I^{\C^{D^t \times D^t}} \otimes \cH_k) \circ 
\ddcG^{\otimes t} \circ \cdots \circ \ddcG^{\otimes t}
(\I^{\C^{D^t \times D^t}} \otimes \cH_1)
)((\gamma'_0)^{\perp})
} + f_0  \\
&    &
{} +
\inprod{\delta_0}{
(
(\I^{\C^{D^t \times D^t}} \otimes \cH_k) \circ 
\ddcG^{\otimes t} \circ \cdots \circ \ddcG^{\otimes t}
(\I^{\C^{D^t \times D^t}} \otimes \cH_1)
)((\gamma'_0)^{\|'})
} + 
e_0 \\
& =  &
\inprod{\delta_0}{
(
(\I^{\C^{D^t \times D^t}} \otimes \cH_k) \circ 
\ddcG^{\otimes t} \circ \cdots \circ \ddcG^{\otimes t}
(\I^{\C^{D^t \times D^t}} \otimes \cH_2)
)(\gamma_1)
} \\
&    &
{} +
\inprod{\delta_0}{
(
(\I^{\C^{D^t \times D^t}} \otimes \cH_k) \circ 
\ddcG^{\otimes t} \circ \cdots \circ \ddcG^{\otimes t}
(\I^{\C^{D^t \times D^t}} \otimes \cH_1)
)((\gamma'_0)^{\|'})
} + 
e_0 + f_0 \\
& =  &
\inprod{\delta_0}{
(\I^{\C^{D^t \times D^t}} \otimes \cH_k)((\gamma'_{k-1})^{\perp})
} \\
&    &
{} +
\sum_{0 \leq i \leq k-1}
\inprod{\delta_0}{
(
(\I^{\C^{D^t \times D^t}} \otimes \cH_k) \circ 
\ddcG^{\otimes t} \circ \cdots \circ \ddcG^{\otimes t}
(\I^{\C^{D^t \times D^t}} \otimes \cH_{i+1})
)((\gamma'_i)^{\|'})
} \\
&   &
{} + 
\sum_{0 \leq i \leq k-1}
(e_i + f_i) \\
& =  &
\inprod{\delta_0}{
(\I^{\C^{D^t \times D^t}} \otimes \cH_k)((\gamma'_{k-1})^{\perp})
} \\
&    &
{} +
\sum_{0 \leq i < j \leq k}
\inprod{(\delta'_{k-j})^{\|'}}{
(
(\I^{\C^{D^t \times D^t}} \otimes \cH_j) \circ 
\ddcG^{\otimes t} \circ \cdots \circ \ddcG^{\otimes t}
(\I^{\C^{D^t \times D^t}} \otimes \cH_{i+1})
)((\gamma'_i)^{\|'})
} \\
&    &
{} +
\sum_{0 \leq i \leq k-1}
\inprod{(\delta'_{k-i-1})^{\perp}}{
(
(\I^{\C^{D^t \times D^t}} \otimes \cH_{i+1})
)((\gamma'_i)^{\|'})
} \\
&   &
{} + 
\sum_{0 \leq i < k}
(e_i + f_i) +
\sum_{0 \leq i < j \leq k}
(g_{ji} + h_{ji}) \\
& =  &
\inprod{\delta_0}{
(\I^{\C^{D^t \times D^t}} \otimes \cH_k)((\gamma'_{k-1})^{\perp})
} \\
&    &
{} +
\sum_{0 \leq i < i+1 < j \leq k}
\inprod{(\delta'_{k-j})^{\|}}{
(
(\I^{\C^{D^t \times D^t}} \otimes \cH_j) \circ 
\ddcG^{\otimes t} \circ \cdots \circ \ddcG^{\otimes t}
(\I^{\C^{D^t \times D^t}} \otimes \cH_{i+1})
)((\gamma'_i)^{\|'})
} \\
&    &
{} +
\sum_{0 \leq i < k}
\inprod{(\delta'_{k-i-1})^{\|}}{
(
(\I^{\C^{D^t \times D^t}} \otimes \cH_{i+1})
)((\gamma'_i)^{\|})
} \\
&    &
{} +
\sum_{0 \leq i \leq k-1}
\inprod{(\delta'_{k-i-1})^{\perp}}{
(
(\I^{\C^{D^t \times D^t}} \otimes \cH_{i+1})
)((\gamma'_i)^{\|})
} \\
&   &
{} + 
\sum_{0 \leq i < k}
(e_i + f_i + d_i +l_i)  +
\sum_{0 \leq i < j \leq k}
(g_{ji} + h_{ji} +  m_{ji}) \\
& =  &
\inprod{\delta_0}{
(\I^{\C^{D^t \times D^t}} \otimes \cH_k)((\gamma'_{k-1})^{\perp})
} \\
&    &
{} +
\sum_{0 \leq i < i+1 < j \leq k}
\inprod{(\delta'_{k-j})^{\|}}{
(\ddcG^{\otimes t} \circ 
(\I^{\C^{D^t \times D^t}} \otimes \cH_{j-1}) 
\circ \cdots \circ 
\ddcG^{\otimes t} \circ
(\I^{\C^{D^t \times D^t}} \otimes \cH_{i+1})
)
((\gamma'_i)^{\|'})
} \\
&    &
{} +
\sum_{0 \leq i \leq k-1}
\inprod{(\delta'_{k-i-1})^{\|}}{(\gamma'_i)^{\|})} \\
&    &
{} + 0 + 
\sum_{0 \leq i < k}
(e_i + f_i + d_i + l_i)  +
\sum_{0 \leq i < j \leq k}
(g_{ji} + h_{ji} + m_{ji}).
\end{eqnarray*}
We will now bound each of the six terms in the last equality.

We start by bounding the fifth term as follows. 
\[
\sum_{0 \leq i < k} |e_i| \leq
2\sqrt{\frac{t(t-1)}{dd'}}
\sum_{0 \leq i < k} \elltwo{\gamma_i} \leq
2\sqrt{\frac{t(t-1)}{dd'}}
\sum_{0 \leq i < k} \lambda_2^i \leq
4\sqrt{\frac{t(t-1)}{dd'}}.
\]
Similarly,
\[
\sum_{0 \leq i < k} |f_i| \leq
4\sqrt[4]{\frac{t(t-1)}{dd'}}.
\]
Next,
\[
\sum_{0 \leq i < k} |d_i| \leq
\sqrt[4]{\frac{t(t-1)}{dd'}}
\sum_{0 \leq i < k} 
\left(
\elltwo{(\gamma'_i)^{\|'}}^2 + \elltwo{(\delta'_{k-i-1})^{\perp}}^2
\right) \leq
3 \sqrt[4]{\frac{t(t-1)}{dd'}}.
\]
Similarly,
\[
\sum_{0 \leq i < k} |l_i| \leq
4\sqrt{\frac{t(t-1)}{dd'}}.
\]
Hence,
\[
\sum_{0 \leq i < k} (|e_i| + |f_i| + |d_i| + |l_i|) \leq
15 \sqrt[4]{\frac{t(t-1)}{dd'}}.
\]

We now bound the sixth term as follows.
\[
\sum_{0 \leq i < j \leq k} |g_{ji}| \leq
2 \sqrt{\frac{t(t-1)}{dd'}}
\sum_{0 \leq i < j \leq k} 
\elltwo{\delta_{k-j}} \cdot \elltwo{\gamma_i} \leq
4 \sqrt{\frac{t(t-1)}{dd'}}
\sum_{0 \leq i < k} \elltwo{\gamma_i} \leq
8 \sqrt{\frac{t(t-1)}{dd'}}.
\]
Similarly,
\[
\sum_{0 \leq i < j \leq k} |h_{ji}| \leq
8 \sqrt[4]{\frac{t(t-1)}{dd'}}, 
~~~~
\sum_{0 \leq i < j \leq k} |m_{ji}| \leq
8 \sqrt{\frac{t(t-1)}{dd'}}.
\]
Hence,
\[
\sum_{0 \leq i < j \leq k} 
(|g_{ji}| + |h_{ji}| + |m_{ji}|) \leq 
24 \sqrt[4]{\frac{t(t-1)}{dd'}}. 
\]

Next we bound the first term as follows.
Observe that
\begin{eqnarray*}
\lefteqn{
|
\inprod{\delta_0}{
(\I^{\C^{D^t \times D^t}} \otimes \cH_k)((\gamma'_{k-1})^{\perp})
}
|
} \\
& \leq &
\elltwo{
(\I^{\C^{D^t \times D^t}} \otimes \cH_k)((\gamma'_{k-1})^{\perp})
} 
\;\leq\;
\lambda_2 \elltwo{(\gamma'_{k-1})^{\perp}} 
\;\leq\;
\lambda_2 \elltwo{(\gamma'_{k-1})^{\perp'}} \;\leq\;
\lambda_2 \elltwo{\gamma'_{k-1}} 
\;\leq\;
\lambda_2^k.
\end{eqnarray*}

We now bound the third term as follows. The proof is very similar
to that of \cite[Lemma~16]{BT}. We give it below for completeness.
\begin{eqnarray*}
\lefteqn{
\sum_{0 \leq i \leq k-1}
\inprod{(\delta'_{k-i-1})^{\|}}{(\gamma'_{i})^{\|}}
} \\
& \leq &
\sum_{0 \leq i \leq k-1}
\elltwo{(\delta'_{k-i-1})^{\|}} \cdot
\elltwo{(\gamma'_{i})^{\|}} \\
&   =  &
\lambda_2^{k-1}
\sum_{0 \leq i \leq k-1}
\lambda_2^{-(k-i-1)}\elltwo{(\delta'_{k-i-1})^{\|}} \cdot
\lambda_2^{-i}\elltwo{(\gamma'_{i})^{\|}} \\
& \leq &
\frac{\lambda_2^{k-1}}{2}
\left(
\sum_{0 \leq i \leq k-1}
\lambda_2^{-2i}\elltwo{(\delta'_{i})^{\|}}^2 +
\sum_{0 \leq i \leq k-1}
\lambda_2^{-2i}\elltwo{(\gamma'_{i})^{\|}}^2 
\right).
\end{eqnarray*}
Now,
\[
\sum_{0 \leq i \leq k-1}
\lambda_2^{-2i}\elltwo{(\delta'_{i})^{\|}}^2
\leq
\sum_{0 \leq i \leq k-1}
\lambda_2^{-2i}\elltwo{(\delta'_{i})^{\|'}}^2 
\left(
1 + 2 \sqrt{\frac{t(t-1)}{dd'}}
\right)^2.
\]
By induction, we observe that
\begin{eqnarray*}
\lefteqn{
\sum_{0 \leq i \leq k-1}
\lambda_2^{-2i}\elltwo{(\delta'_{i})^{\|'}}^2
} \\
& \leq &
\sum_{0 \leq i \leq k-1}
\lambda_2^{-2i}\elltwo{(\delta'_{i})^{\|'}}^2 + 
\lambda_2^{-2(k-1)}\elltwo{(\delta'_{k-1})^{\perp'}}^2 
\;=\;
\sum_{0 \leq i \leq k-2}
\lambda_2^{-2i}\elltwo{(\delta'_{i})^{\|'}}^2 +
\lambda_2^{-2(k-1)}\elltwo{\delta'_{k-1}}^2 \\
& \leq &
\sum_{0 \leq i \leq k-2}
\lambda_2^{-2i}\elltwo{(\delta'_{i})^{\|'}}^2 +
\lambda_2^{-2(k-2)}\elltwo{(\delta'_{k-2})^{\perp'}}^2 
\;\leq\;
\elltwo{\delta'_{0}}^2 
\;\leq\;
\elltwo{\delta_{0}}^2 
\; = \;
1.
\end{eqnarray*}
Thus,
\[
\sum_{0 \leq i \leq k-1}
\lambda_2^{-2i}\elltwo{(\delta'_{i})^{\|}}^2
\leq
\left(
1 + 2 \sqrt{\frac{t(t-1)}{dd'}}
\right)^2.
\]
Similarly,
\[
\sum_{0 \leq i \leq k-1}
\lambda_2^{-2i}\elltwo{(\gamma'_{i})^{\|}}^2 \leq 
\left(
1 + 2 \sqrt{\frac{t(t-1)}{dd'}}
\right)^2.
\]
Thus,
\[
\sum_{0 \leq i \leq k-1}
\inprod{(\delta'_{k-i-1})^{\|}}{(\gamma'_{i})^{\|}}
 \leq
\lambda_2^{k-1}
\left(
1 + 2 \sqrt{\frac{t(t-1)}{dd'}}
\right)^2.
\]

\begin{comment}
The third term can be bounded in a similar fashion as in the last
section. This gives us
\[
|\inprod{(\delta'_{k-i-2})^{\|}}{\ddcG^{\otimes t}((\gamma'_i)^{\|})}|
\leq
\left(\lambda_1 + 3\frac{t(t-1)}{d}\right)
\elltwo{(\delta'_{k-i-2})^{\|'}} \cdot \elltwo{(\gamma'_{i})^{\|'}}.
\]
(Though the unitary $\dcG$ is sensitive only to the contents of
$\C^d$, for the above inequality we can pretend that it is sensitive
to the entire contents of $\C^d \otimes \C^{d'}$ without any harm.)
This gives
\begin{eqnarray*}
\sum_{0 \leq i \leq k-2}
|\inprod{(\delta'_{k-i-2})^{\|}}{\ddcG^{\otimes t}((\gamma'_i)^{\|})}|
& \leq &
\left(\lambda_1 + 3\frac{t(t-1)}{d}\right)
\sum_{0 \leq i \leq k-2}
\elltwo{(\delta'_{k-i-2})^{\|'}} \cdot \elltwo{(\gamma'_{i})^{\|'}} \\
& \leq &
\lambda_2^{k-2} \left(\lambda_1 + 3\frac{t(t-1)}{d}\right),
\end{eqnarray*}
where we bounded the summation by $\lambda_2^{k-2}$ as in the upper bound
of 
$
\lambda_2^{k-1} 
\left(
1 + 2 \sqrt{\frac{t(t-1)}{dd'}}
\right)^2
$
proved just above.
This gives
\[
\sum_{0 \leq i \leq k-2}
|\inprod{(\delta'_{k-i-2})^{\|}}{\ddcG^{\otimes t}((\gamma'_i)^{\|})}|
\leq
\lambda_1 \lambda_2^{k-2} +
3 \frac{t(t-1)}{d}.
\]
\end{comment}

We now bound the second term as follows: Fix
\[
(\C^{D^t \times D^t} \otimes W'_2) \cap (W')^\perp \ni
(\gamma'_i)^{\|'} =
\sum_\sigma b_\sigma 
(\Sigma^{(\C^D)^{\otimes t}} \otimes \Sigma^{(\C^{dd'})^{\otimes t}})
(\beta_1)_\sigma \otimes (\alpha'_2),
\]
\[
(\C^{D^t \times D^t} \otimes W'_2) \cap (W')^\perp \ni
(\delta'_{k-j})^{\|'} =
\sum_\sigma c_\sigma 
(\Sigma^{(\C^D)^{\otimes t}} \otimes \Sigma^{(\C^{dd'})^{\otimes t}})
(\theta_1)_\sigma \otimes (\alpha'_2),
\]
where 
$\elltwo{(\beta_1)_\sigma} = 1$,
$\elltwo{(\theta_1)_\sigma} = 1$,
$b_\sigma, c_\sigma \in \C$.
Let $b$ be the $t!$-tuple whose $\sigma$th entry
is $|b_\sigma|$.
Let $c$ be the $t!$-tuple whose $\sigma$th entry
is $|c_\sigma|$.
By Lemma~\ref{lem:perm},
\[
\elltwo{(\gamma'_i)^{\|'}}^2 = 
\sum_\sigma |b_\sigma|^2 \elltwo{(\alpha'_2)_\alpha}^2 \geq
\elltwo{b}^2 \left(1 - \frac{t(t-1)}{2dd'}\right),
\]
which gives
$
\elltwo{b}^2 \leq 
\frac{\elltwo{(\gamma'_i)^{\|'}}^2}{1 - \frac{t(t-1)}{2dd'}}.
$
Similarly,
$
\elltwo{c}^2 \leq 
\frac{\elltwo{(\delta'_{k-j})^{\|'}}^2}{1 - \frac{t(t-1)}{2dd'}}.
$
Define the matrix
\[
\C^{D^t \times D^t} \otimes W_2 \ni
(\delta'_{k-j})^{\|} :=
\sum_\sigma c_\sigma 
(\Sigma^{(\C^D)^{\otimes t}} \otimes \Sigma^{(\C^{dd'})^{\otimes t}})
((\theta_1)_\sigma \otimes (\alpha_2)).
\]
By Lemma~\ref{lem:subspaces},
\[
\elltwo{(\delta'_{k-j})^{\|'} - (\delta'_{k-j})^{\|}} \leq 
2\sqrt{\frac{t(t-1)}{dd'}} \elltwo{(\delta'_{k-j})^{\|'}}.
\]
Moreover, as $(\gamma'_i)^{\|'} \in (W')^\perp$, 
\[
0 =
\inprod{(\gamma'_i)^{\|'}}{
(\Sigma^{(\C^D)^{\otimes t}} \otimes \Sigma^{(\C^{dd'})^{\otimes t}})
(\alpha_1 \otimes \alpha'_2)
} =
b_\sigma^*
\inprod{(\beta_1)_\sigma}{\alpha_1} \cdot
\inprod{\alpha'_2}{\alpha'_2},
\]
implying that
$
\inprod{(\beta_1)_\sigma}{\alpha_1} = 0
$
for all $\sigma \in S_t$. 

Observe that
\begin{eqnarray*}
\lefteqn{
|
\inprod{(\delta'_{k-j})^{\|}}{
(\ddcG^{\otimes t} \circ 
(\I^{\C^{D^t \times D^t}} \otimes \cH_{j-1}) 
\circ \cdots \circ 
\ddcG^{\otimes t} \circ
(\I^{\C^{D^t \times D^t}} \otimes \cH_{i+1})
)
((\gamma'_i)^{\|'})
}
|
} \\
& = &
\left|
\sum_{\sigma \neq \sigma'} c^*_{\sigma'} b_\sigma
\langle
((\Sigma')^{(\C^D)^{\otimes t}} \otimes (\Sigma')^{(\C^{dd'})^{\otimes t}})
((\theta_1)_{\sigma'} \otimes \alpha_2), 
\right. \\
&   &
~~~~~
(\Sigma^{(\C^D)^{\otimes t}} \otimes \Sigma^{(\C^{dd'})^{\otimes t}})
(
(
\ddcG^{\otimes t} \circ 
(\I^{\C^{D^t \times D^t}} \otimes \cH_{j-1}) 
\circ \cdots \circ 
\ddcG^{\otimes t} \circ
(\I^{\C^{D^t \times D^t}} \otimes \cH_{i+1})
)
((\beta_1)_\sigma \otimes \alpha'_2)
) \rangle \\
&    &
{} +
\sum_{\sigma} c^*_{\sigma} b_\sigma
\langle
(\Sigma^{(\C^D)^{\otimes t}} \otimes \Sigma^{(\C^{dd'})^{\otimes t}})
((\theta_1)_{\sigma} \otimes \alpha_2), \\
&   &
~~~~~
\left.
(\Sigma^{(\C^D)^{\otimes t}} \otimes \Sigma^{(\C^{dd'})^{\otimes t}})
(
(
\ddcG^{\otimes t} \circ 
(\I^{\C^{D^t \times D^t}} \otimes \cH_{j-1}) 
\circ \cdots \circ 
\ddcG^{\otimes t} \circ
(\I^{\C^{D^t \times D^t}} \otimes \cH_{i+1})
)
((\beta_1)_\sigma \otimes \alpha'_2)
) \rangle
\right| \\
& \leq &
\sum_{\sigma \neq \sigma'} |c_{\sigma'}| |b_\sigma|
\left|
\langle
((\Sigma')^{(\C^D)^{\otimes t}} \otimes (\Sigma')^{(\C^{dd'})^{\otimes t}})
((\theta_1)_{\sigma'} \otimes \alpha_2), 
\right. \\
&   &
~~~~~
\left.
(\Sigma^{(\C^D)^{\otimes t}} \otimes \Sigma^{(\C^{dd'})^{\otimes t}})
(
(
\ddcG^{\otimes t} \circ 
(\I^{\C^{D^t \times D^t}} \otimes \cH_{j-1}) 
\circ \cdots \circ 
\ddcG^{\otimes t} \circ
(\I^{\C^{D^t \times D^t}} \otimes \cH_{i+1})
)
((\beta_1)_\sigma \otimes \alpha'_2)
) \rangle
\right| \\
&    &
{} +
\sum_{\sigma} |c_{\sigma}| |b_\sigma|
\left|
\langle
(\Sigma^{(\C^D)^{\otimes t}} \otimes \Sigma^{(\C^{dd'})^{\otimes t}})
((\theta_1)_{\sigma} \otimes \alpha_2), 
\right. \\
&   &
~~~~~
\left.
(\Sigma^{(\C^D)^{\otimes t}} \otimes \Sigma^{(\C^{dd'})^{\otimes t}})
(
(
\ddcG^{\otimes t} \circ 
(\I^{\C^{D^t \times D^t}} \otimes \cH_{j-1}) 
\circ \cdots \circ 
\ddcG^{\otimes t} \circ
(\I^{\C^{D^t \times D^t}} \otimes \cH_{i+1})
)
((\beta_1)_\sigma \otimes \alpha'_2)
) \rangle
\right|.
\end{eqnarray*}

Fix a $\sigma \in S_t$. 
Let $1 \leq l \leq k-1$ be a positive integer  
Fix a $\frac{\epsilon}{8 t d^k}$-good sequence of unitary 
maps $(U_l, \ldots, U_l)$ on
$\C^d \otimes \C^{d'} \cong \C^{dd'}$. 
Let the corresponding unitary superoperators 
tensored with the identity superoperator on $\C^{D \times D}$
be denoted by $\dU_l, \ldots, \dU_1$. 
Let $x_0 := e_{\vec{i_0}} \otimes e'_{\vec{j_0}}$ denote a computational
basis vector of
$
(\C^d)^{\otimes t} \otimes (\C^{d'})^{\otimes t} \cong \C^{(dd')^t}
$.
We now calculate the matrix
\begin{eqnarray*}
\lefteqn{
(
\ddcG^{\otimes t} \circ \dU_l^{\otimes t} 
\circ \cdots \circ 
\ddcG^{\otimes t} \circ
\dU_1^{\otimes t }
)
(
(\Sigma^{(\C^D)^{\otimes t}} \otimes \Sigma^{(\C^{dd'})^{\otimes t}})
((\beta_1)_\sigma \otimes (x_0 x_0^\dag))
)
} \\
& = &
(\Sigma^{(\C^D)^{\otimes t}} \otimes \Sigma^{(\C^{dd'})^{\otimes t}})
(
(
\ddcG^{\otimes t} \circ \dU_l^{\otimes t} 
\circ \cdots \circ 
\ddcG^{\otimes t} \circ
\dU_1^{\otimes t} 
)
((\beta_1)_\sigma \otimes (x_0 x_0^\dag))
)
\end{eqnarray*}
Fix two sequences of computational basis vectors 
$(e_{\vec{i_l}}, \ldots, e_{\vec{i_1}})$,
$(e_{\vec{i'_l}}, \ldots, e_{\vec{i'_1}})$ of $\C^d$.
Starting from computational basis vector $x_0$ of $\C^{(dd')^t}$,
define the two sequences of vectors 
$(x_l, \ldots, x_1)$, $(x'_l, \ldots, x'_1)$ of $\C^{(dd')^t}$
accordingly, as in Section~\ref{subsec:epsgood}. Let 
$p(\vec{i_l}, \ldots, \vec{i_1})$,
$p(\vec{i'_l}, \ldots, \vec{i'_1})$ denote the probabilities of
obtaining the corresponding sequences of outcomes on measuring 
$\C^{d^t}$ in its computational basis. 
For a computational basis vector $\vec{i}$ of $\C^{d^t}$, let 
$V^{\otimes \vec{i}}$ be the corresponding tensor product of 
unitary operators on $\C^D$ arising from the $1$-qTPE $\cG$ 
(note that the unitary operators of $\cG$ are indexed by the computational
basis vectors of $\C^d$). Let 
$
(
V^{\otimes \vec{i_l}}, \ldots, V^{\otimes \vec{i_1}}, 
V^{\otimes \vec{i_0}}
),
$
$
(
(V^\dag)^{\otimes \vec{i'_l}}, \ldots, (V^\dag)^{\otimes \vec{i'_1}}, 
(V^\dag)^{\otimes \vec{i_0}}
),
$
be corresponding sequences of unitary maps on $(\C^D)^{\otimes t}$.
Then 
\begin{eqnarray*}
\lefteqn{
(
\ddcG^{\otimes t} \circ \dU_l^{\otimes t} 
\circ \cdots \circ 
\ddcG^{\otimes t} \circ
\dU_1^{\otimes t} 
)
((\beta_1)_\sigma \otimes (x_0 x_0^\dag))
} \\
& = &
\sum_{
\begin{array}{c}
(\vec{i_l}, \ldots, \vec{i_1}) \\
(\vec{i'_l}, \ldots, \vec{i'_1})
\end{array}
}
\sqrt{
p(\vec{i_l}, \ldots, \vec{i_1})
p(\vec{i'_l}, \ldots, \vec{i'_1})
} \\
&   &
~~~~~~~~~
(
(V^{\otimes \vec{i_l}} \cdots V^{\otimes \vec{i_1}})
(\beta_1)_\sigma
((V^\dag)^{\otimes \vec{i'_1}} \cdots (V^\dag)^{\otimes \vec{i'_l}}
)
) \otimes 
(x_l (x'_l)^\dag).
\end{eqnarray*}

Let $\sigma' \in S_t$, $\sigma' \neq \sigma$. 
Suppose all the $t$ entries of $(\vec{i_0}, \vec{j_0})$ are distinct.
Then 
\begin{eqnarray*}
\lefteqn{
|
\langle
((\Sigma')^{(\C^D)^{\otimes t}} \otimes (\Sigma')^{(\C^{dd'})^{\otimes t}})
((\theta_1)_{\sigma'} \otimes \alpha_2),
} \\
&   &
~~~~~
(\Sigma^{(\C^D)^{\otimes t}} \otimes \Sigma^{(\C^{dd'})^{\otimes t}})
(
(
\ddcG^{\otimes t} \circ \dU_l^{\otimes t} 
\circ \cdots \circ 
\ddcG^{\otimes t} \circ
\dU_1^{\otimes t}
)
((\beta_1)_\sigma \otimes (x_0 x_0^\dag))
) \rangle
| \\
& \leq &
\sum_{
\begin{array}{c}
(\vec{i_l}, \ldots, \vec{i_1}) \\
(\vec{i'_l}, \ldots, \vec{i'_1})
\end{array}
}
\sqrt{
p(\vec{i_l}, \ldots, \vec{i_1})
p(\vec{i'_l}, \ldots, \vec{i'_1})
} \\
&   &
~~~~~
|
\inprod{(\theta_1)_{\sigma'}}{
((\Sigma')^{-1} \Sigma)^{(\C^{D})^{\otimes t}}
(
(V^{\otimes \vec{i_l}} \cdots V^{\otimes \vec{i_1}})
(\beta_1)_\sigma
((V^\dag)^{\otimes \vec{i'_1}} \cdots (V^\dag)^{\otimes \vec{i'_l}})
)
} 
| \\
&    &
~~~~~~~~
{} \cdot
|
\inprod{\alpha_2}{
((\Sigma')^{-1} \Sigma)^{(\C^{dd'})^{\otimes t}}(x_l (x'_l)^\dag)
}
| \\
& \leq &
\sum_{
\begin{array}{c}
(\vec{i_l}, \ldots, \vec{i_1}) \\
(\vec{i'_l}, \ldots, \vec{i'_1})
\end{array}
}
\sqrt{
p(\vec{i_l}, \ldots, \vec{i_1})
p(\vec{i'_l}, \ldots, \vec{i'_1})
} \; 
|
\inprod{\alpha_2}{
((\Sigma')^{-1} \Sigma)^{(\C^{dd'})^{\otimes t}}(x_l (x'_l)^\dag)
}
| \\
&   =  &
(dd')^{-t/2}
\sum_{
\begin{array}{c}
(\vec{i_l}, \ldots, \vec{i_1}) \\
(\vec{i'_l}, \ldots, \vec{i'_1})
\end{array}
}
\sqrt{
p(\vec{i_l}, \ldots, \vec{i_1})
p(\vec{i'_l}, \ldots, \vec{i'_1})
} \; 
|
\inprod{x'_l}{
((\Sigma')^{-1} \Sigma)^{(\C^{dd'})^{\otimes t}}(x_l)
}
|.
\end{eqnarray*}
Let $\{\} \neq T \subseteq [t]$ be the set of coordinates are 
not mapped to themselves by the permutation
$(\sigma')^{-1} \sigma$. Let $T'$ be a subset of $[t]$ denoting
the coordinates where  sequence $(\vec{i_l}, \ldots,  \vec{i_1})$
disagrees with the sequence $(\vec{i'_l}, \ldots,  \vec{i'_1})$. 
Define $T'' := [t] \setminus (T \cup T')$. We use the notation
$(\vec{i_l}, \ldots,  \vec{i_1})_{T''}$ to denote the sequence
restricted to the coordinates in $T''$, and the notation
$
(\vec{i_l}, \ldots, \vec{i_1})_{T \cup T'} 
\not{\equiv} (\vec{i'_l}, \ldots, \vec{i'_1})_{T \cup T'}
$
to denote that the two sequences disagree on every coordinate in 
$T'$. Thus, we can bound the above quantity by
\begin{eqnarray*}
\lefteqn{
|
\langle
((\Sigma')^{(\C^D)^{\otimes t}} \otimes (\Sigma')^{(\C^{dd'})^{\otimes t}})
((\theta_1)_{\sigma'} \otimes \alpha_2),
} \\
&   &
~~~~~
(\Sigma^{(\C^D)^{\otimes t}} \otimes \Sigma^{(\C^{dd'})^{\otimes t}})
(
(
\ddcG^{\otimes t} \circ \dU_l^{\otimes t} 
\circ \cdots \circ 
\ddcG^{\otimes t} \circ
\dU_1^{\otimes t} 
)
((\beta_1)_\sigma \otimes (x_0 x_0^\dag))
) \rangle
| \\
& \leq &
(dd')^{-t/2}
\sum_{T'': T'' \cap T = \{\}}
\sum_{(\vec{i_l}, \ldots, \vec{i_1})_{T''}}
\sum_{
\begin{array}{c}
(\vec{i_l}, \ldots, \vec{i_1})_{T \cup T'} \\
{} \not{\equiv} (\vec{i'_l}, \ldots, \vec{i'_1})_{T \cup T'}
\end{array}
} \\
&   &
~~~~~
p((\vec{i_l}, \ldots, \vec{i_1})_{T''})
\sqrt{
p(
(\vec{i_l}, \ldots, \vec{i_1})_{T \cup T''} | 
(\vec{i_l}, \ldots, \vec{i_1})_{T''}
)
p(
(\vec{i'_l}, \ldots, \vec{i'_1})_{T \cup T''} | 
(\vec{i_l}, \ldots, \vec{i_1})_{T''}
)
} \\
&   &
~~~~~~~~
|
\inprod{x'_l}{
((\Sigma')^{-1} \Sigma)^{(\C^{dd'})^{\otimes t}}(x_l)
}
| \\
& \leq &
(dd')^{-t/2}
\sum_{T'': T'' \cap T = \{\}}
\sum_{(\vec{i_l}, \ldots, \vec{i_1})_{T''}}
p((\vec{i_l}, \ldots, \vec{i_1})_{T''}) \,
d^{l |T \cup T'|}
\left(\frac{\epsilon}{t d^k}\right)^{|T \cup T'|} \\
&   =  &
(dd')^{-t/2}
\sum_{T'': T'' \cap T = \{\}}
(t^{-1} \epsilon)^{|[t] \setminus T''|} 
\;  = \;
(dd')^{-t/2}
(t^{-1} \epsilon)^{|T|} (1 + t^{-1} \epsilon)^{t - |T|} \\
& \leq &
(dd')^{-t/2} (t^{-1} \epsilon)^{t - f_{(\sigma')^{-1} \sigma}}
e^\epsilon 
\;\leq\;
2 (dd')^{-t/2} (t^{-1} \epsilon)^{t - f_{(\sigma')^{-1} \sigma}},
\end{eqnarray*}
where $f_{(\sigma')^{-1} \sigma}$ denotes the number of fixed
points of the permuation $(\sigma')^{-1} \sigma$ and the second
inequality follows from the fact that the sequence $(U_l, \ldots, U_1)$ is
$\frac{\epsilon}{8 t d^k}$-good and so it maps orthogonal states to
almost orthogonal states.

When $\sigma' = \sigma$, the set of unfixed points $T = \{\}$. 
Hence we can upper bound
the corresponding quantity by
\begin{eqnarray*}
\lefteqn{
|
\langle
((\Sigma)^{(\C^D)^{\otimes t}} \otimes (\Sigma)^{(\C^{dd'})^{\otimes t}})
((\theta_1)_{\sigma} \otimes \alpha_2),
} \\
&   &
~~~~~
(\Sigma^{(\C^D)^{\otimes t}} \otimes \Sigma^{(\C^{dd'})^{\otimes t}})
(
(
\ddcG^{\otimes t} \circ \dU_l^{\otimes t} 
\circ \cdots \circ 
\ddcG^{\otimes t} \circ
\dU_1^{\otimes t}
)
((\beta_1)_\sigma \otimes (x_0 x_0^\dag))
) \rangle
| \\
&   =  &
\left|
\sum_{
\begin{array}{c}
(\vec{i_l}, \ldots, \vec{i_1}) \\
(\vec{i'_l}, \ldots, \vec{i'_1})
\end{array}
}
\sqrt{
p(\vec{i_l}, \ldots, \vec{i_1})
p(\vec{i'_l}, \ldots, \vec{i'_1})
} 
\right. \\
&   &
~~~~~~~~~~~
\left.
\inprod{(\theta_1)_{\sigma}}{
(V^{\otimes \vec{i_l}} \cdots V^{\otimes \vec{i_1}})
(\beta_1)_\sigma
((V^\dag)^{\otimes \vec{i'_1}} \cdots (V^\dag)^{\otimes \vec{i'_l}})
} 
\inprod{\alpha_2}{x_l (x'_l)^\dag}
\right| \\
&   =  &
(dd')^{-t/2}
\left|
\sum_{
\begin{array}{c}
(\vec{i_l}, \ldots, \vec{i_1}) \\
(\vec{i'_l}, \ldots, \vec{i'_1})
\end{array}
}
\sqrt{
p(\vec{i_l}, \ldots, \vec{i_1})
p(\vec{i'_l}, \ldots, \vec{i'_1})
} 
\right. \\
&   &
~~~~~~~~~~~~~~~~~~~~~~~~
\left.
\inprod{(\theta_1)_{\sigma}}{
(V^{\otimes \vec{i_l}} \cdots V^{\otimes \vec{i_1}})
(\beta_1)_\sigma
((V^\dag)^{\otimes \vec{i'_1}} \cdots (V^\dag)^{\otimes \vec{i'_l}})
} 
\inprod{x'_l}{x_l}
\right| \\
& \leq &
(dd')^{-t/2}
\left|
\sum_{
(\vec{i_l}, \ldots, \vec{i_1}) 
}
p(\vec{i_l}, \ldots, \vec{i_1}) 
\right. \\
&   &
~~~~~~~~~~~~~~~~~~~~~~~~
\left.
\inprod{(\theta_1)_{\sigma}}{
(V^{\otimes \vec{i_l}} \cdots V^{\otimes \vec{i_1}})
(\beta_1)_\sigma
((V^\dag)^{\otimes \vec{i_1}} \cdots (V^\dag)^{\otimes \vec{i_l}})
} 
\inprod{x_l}{x_l}
\right| \\
&      &
{} +
(dd')^{-t/2}
\sum_{T'': T'' \neq [t]}
\sum_{(\vec{i_l}, \ldots, \vec{i_1})_{T''}}
\sum_{
\begin{array}{c}
(\vec{i_l}, \ldots, \vec{i_1})_{[t] \setminus T''} \\
{} \not{\equiv} (\vec{i'_l}, \ldots, \vec{i'_1})_{[t] \setminus T''}
\end{array}
} \\
&   &
~~~~~~~~~~
p((\vec{i_l}, \ldots, \vec{i_1})_{T''})
\sqrt{
p(
(\vec{i_l}, \ldots, \vec{i_1})_{[t] \setminus T''} | 
(\vec{i_l}, \ldots, \vec{i_1})_{T''}
)
p(
(\vec{i'_l}, \ldots, \vec{i'_1})_{[t] \setminus T''} | 
(\vec{i_l}, \ldots, \vec{i_1})_{T''}
)
} \\
&   &
~~~~~~~~~~~~~
|\inprod{x'_l}{x_l}| \\
& \leq &
(dd')^{-t/2}
\left|
\sum_{
(\vec{i_l}, \ldots, \vec{i_1}) 
}
p(\vec{i_l}, \ldots, \vec{i_1}) 
\right. \\
&   &
~~~~~~~~~~~~~~~~~~~~~~~~
\left.
|
\inprod{(\theta_1)_{\sigma}}{
(V^{\otimes \vec{i_l}} \cdots V^{\otimes \vec{i_1}})
(\beta_1)_\sigma
((V^\dag)^{\otimes \vec{i_1}} \cdots (V^\dag)^{\otimes \vec{i_l}})
} 
\right| \\
&   &
{} +
(dd')^{-t/2}
\sum_{T'': T'' \neq [t]}
\sum_{(\vec{i_l}, \ldots, \vec{i_1})_{T''}}
p((\vec{i_l}, \ldots, \vec{i_1})_{T''}) \,
d^{l |[t] \setminus T''|}
\left(\frac{\epsilon}{t d^k}\right)^{|[t] \setminus T''|} \\
& \leq &
(dd')^{-t/2}
\left|
\sum_{
(\vec{i_l}, \ldots, \vec{i_1}) 
} 
d^{-lt} 
\inprod{(\theta_1)_{\sigma}}{
(V^{\otimes \vec{i_l}} \cdots V^{\otimes \vec{i_1}})
(\beta_1)_\sigma
((V^\dag)^{\otimes \vec{i_1}} \cdots (V^\dag)^{\otimes \vec{i_l}})
} 
\right| \\
&    &
{} +
(dd')^{-t/2} 
\sum_{
(\vec{i_l}, \ldots, \vec{i_1}) 
} 
| p(\vec{i_l}, \ldots, \vec{i_1}) - d^{-lt}|  \\
&      &
+
(dd')^{-t/2}
\sum_{T'': T'' \neq [t]}
(t^{-1} \epsilon)^{|[t] \setminus T''|} \\
& \leq &
(dd')^{-t/2}
|
\inprod{(\theta_1)_{\sigma}}{
(\cG^{\otimes t})^l
((\beta_1)_\sigma)
}  
| + 
(dd')^{-t/2}
\frac{3 \epsilon}{8 t d^k} l t
 +
(dd')^{-t/2}
((1 + t^{-1} \epsilon)^{t} - 1) \\
& \leq &
(dd')^{-t/2} \elltwo{(\cG^{\otimes t})^l((\beta_1)_\sigma)} +
(dd')^{-t/2} 
\frac{3 k \epsilon}{8 d^k} +
(dd')^{-t/2} (e^{\epsilon} - 1) \\
& \leq &
(dd')^{-t/2} \lambda_1^l +
3 (dd')^{-t/2} \epsilon.
\end{eqnarray*}
Above, we used the fact that the sequence $(U_l, \ldots, U_1)$ is
$\frac{\epsilon}{8 t d^k}$-good in the second and fourth inequalities
which implies that orthogonal states get mapped to almost orthogonal
states and the measurment results are almost uniform. We also
used the triangle inequality for the $\ell_1$-distance in the
fourth inequality. For the sixth inequality, we used the observation that
$\cG^{\otimes t}$ is a $1$-qTPE by Fact~\ref{fact:tensor}, 
Fact~\ref{fact:squaring} and the fact that
$\inprod{(\beta_1)_\sigma}{\alpha_1} = 0$ which was proved earlier.

We now evaluate
\begin{eqnarray*}
\lefteqn{
\sum_{\sigma \neq \sigma'} |c_{\sigma'}| |b_{\sigma}|
|
\langle
((\Sigma')^{(\C^D)^{\otimes t}} \otimes (\Sigma')^{(\C^{dd'})^{\otimes t}})
((\theta_1)_{\sigma'} \otimes \alpha_2),
} \\
&   &
~~~~~
(\Sigma^{(\C^D)^{\otimes t}} \otimes \Sigma^{(\C^{dd'})^{\otimes t}})
(
(
\ddcG^{\otimes t} \circ \dU_l^{\otimes t} 
\circ \cdots \circ 
\ddcG^{\otimes t} \circ
\dU_1^{\otimes t} 
)
((\beta_1)_\sigma \otimes \alpha'_2)
) \rangle
| \\
& \leq &
\sum_{\sigma \neq \sigma'} |c_{\sigma'}| |b_{\sigma}|
2 (t^{-1} \epsilon)^{t - f_{(\sigma')^{-1} \sigma}} \\
&    = &
2 c^\dagger N b 
\;\leq\;
2 \ellinfty{N} \elltwo{c} \elltwo{b}  \\
& \leq &
8 \epsilon^2 
\elltwo{(\gamma'_i)^{\|'}} \elltwo{(\delta'_{k-j})^{\|'}},
\end{eqnarray*}
where we used Lemma~\ref{lem:fixedpoints} in the second inequality.
Similarly,
\begin{eqnarray*}
\lefteqn{
\sum_{\sigma} |c_{\sigma}| |b_{\sigma}|
|
\langle
((\Sigma)^{(\C^D)^{\otimes t}} \otimes (\Sigma)^{(\C^{dd'})^{\otimes t}})
((\theta_1)_{\sigma} \otimes \alpha_2),
} \\
&   &
~~~~~
(\Sigma^{(\C^D)^{\otimes t}} \otimes \Sigma^{(\C^{dd'})^{\otimes t}})
(
(
\ddcG^{\otimes t} \circ \dU_l^{\otimes t} 
\circ \cdots \circ 
\ddcG^{\otimes t} \circ
\dU_1^{\otimes t} 
)
((\beta_1)_\sigma \otimes \alpha'_2)
) \rangle
| \\
& \leq &
\sum_{\sigma} |c_{\sigma}| |b_{\sigma}|
(\lambda_1^l + 3 \epsilon) \\
& \leq &
(\lambda_1^l + 3 \epsilon) \elltwo{c} \elltwo{b}  \\
& \leq &
2 (\lambda_1^l + 3 \epsilon)
\elltwo{(\gamma'_i)^{\|'}} \elltwo{(\delta'_{k-j})^{\|'}}
\end{eqnarray*}

We now let $(U_l, \ldots, U_1)$ range over unitaries from the sequence 
of qTPEs
$(\cH_{i+l}, \ldots, \cH_{i+1})$ which is assumed to be
$\frac{\epsilon}{8 t d^k}$-good
by Lemma~\ref{lem:epsgood}. This gives us
\begin{eqnarray*}
\lefteqn{
|
\inprod{(\delta'_{k-j})^{\|}}{
(\ddcG^{\otimes t} \circ 
(\I^{\C^{D^t \times D^t}} \otimes \cH_{j-1}) 
\circ \cdots \circ 
\ddcG^{\otimes t} \circ
(\I^{\C^{D^t \times D^t}} \otimes \cH_{i+1})
)
((\gamma'_i)^{\|'})
}
|
} \\
& \leq &
8 \epsilon^2 
\elltwo{(\gamma'_i)^{\|'}} \elltwo{(\delta'_{k-j})^{\|'}} +
2 (\lambda_1^{j-i-1} + 3 \epsilon)
\elltwo{(\gamma'_i)^{\|'}} \elltwo{(\delta'_{k-j})^{\|'}}.
\end{eqnarray*}
We can now finally bound the second term by
\begin{eqnarray*}
\lefteqn{
\sum_{0 \leq i < i+1 < j \leq k}
|
\inprod{(\delta'_{k-j})^{\|}}{
(\ddcG^{\otimes t} \circ 
(\I^{\C^{D^t \times D^t}} \otimes \cH_{j-1}) 
\circ \cdots \circ 
\ddcG^{\otimes t} \circ
(\I^{\C^{D^t \times D^t}} \otimes \cH_{i+1})
)
((\gamma'_i)^{\|'})
}
|
} \\
& \leq &
8 \epsilon^2 
\sum_{0 \leq i < i+1 < j \leq k}
\elltwo{(\gamma'_i)^{\|'}} \elltwo{(\delta'_{k-j})^{\|'}} \\
&     &
{} +
\sum_{0 \leq i < i+1 < j \leq k}
2 (\lambda_1^{j-i-1} + 3 \epsilon)
\elltwo{(\gamma'_i)^{\|'}} \elltwo{(\delta'_{k-j})^{\|'}} \\
& \leq &
(8 \epsilon^2 + 2(\lambda_1 + 3 \epsilon))
\sum_{0 \leq i < i+1 < j \leq k}
\elltwo{(\gamma'_i)^{\|'}} \elltwo{(\delta'_{k-j})^{\|'}} \\
& \leq &
(2 \lambda_1 + 14 \epsilon)
\sum_{0 \leq i < i+1 < j \leq k}
\elltwo{\gamma'_i} \elltwo{\delta'_{k-j}} \\
& \leq &
2 (\lambda_1 + 7 \epsilon)
\sum_{0 \leq i \leq k-2} \elltwo{\gamma'_i} 
\left(\sum_{j=i+2}^k \lambda_2^{k-j}\right) \\
& \leq &
4 (\lambda_1 + 7 \epsilon)
\sum_{0 \leq i \leq k-2} \elltwo{\gamma'_i} 
\;\leq\;
4 (\lambda_1 + 7 \epsilon)
\sum_{0 \leq i \leq k-2} \lambda_2^i 
\;\leq\;
8 (\lambda_1 + 7 \epsilon).
\end{eqnarray*}

Putting everything together, we have finally shown that
\begin{eqnarray*}
|\inprod{\delta_0}{(\cG \zigzag \vcH)(\gamma_0)}| 
& \leq &
\lambda_2^k + 8(\lambda_1 + 7\epsilon) + 
\lambda_2^{k-1} + 8 \sqrt{\frac{t(t-1)}{dd'}} +
15 \sqrt[4]{\frac{t(t-1)}{dd'}} +
24 \sqrt[4]{\frac{t(t-1)}{dd'}} \\
& \leq &
8(\lambda_1 + 7\epsilon) + \lambda_2^{k-1} + \lambda_2^k + 
47 \sqrt[4]{\frac{t(t-1)}{dd'}}.
\end{eqnarray*}
We have thus shown the following theorem.
\begin{theorem}
Let $s \geq 4$, $d, d' \geq 100$ be integers. Let $k \leq \log s$
be an integer. Let $t$ be an integer such that 
$D \geq dd' \geq 10 t^2$. Let $0 < \epsilon < 10^{-2}$.
Let $\cG = \{V_i\}_{i=1}^d$ be a
$(D, d, \lambda_1, 1)$-qTPE. 
Let $\cH_j = \{U_i(j)\}_{i=1}^s$, $1 \leq j \leq k$ be 
$(dd', s, \lambda_2, t)$-qTPEs such that the sequence
$\vcH := (\cH_k, \ldots, \cH_1)$ is $\frac{\epsilon}{8 t d^k}$-good.
Then $\cG \zigzag \vcH$ is a
$(Ddd', s^k, \lambda, t)$-qTPE where 
\[
\lambda := 
8(\lambda_1 + 7\epsilon) + \lambda_2^{k-1} + \lambda_2^k + 
47 \sqrt[4]{\frac{t(t-1)}{dd'}}.
\]
Moreover for $d' \geq 30 \log s (\log s + \log d) d^{2k+1} \epsilon^{-2}$,
such a sequence $\vcH$ exists with $\lambda_2 < 8 s^{-1/2}$.
\end{theorem}
\paragraph{Remarks:} \ \\

\noindent
1.\ For $t = 1$, we get the bound
$\lambda = 8(\lambda_1 + 7\epsilon) + \lambda_2^{k-1} + \lambda_2^k$
which is the same as the bound in \cite{BT} except for the 
constants involving the $\lambda_1$ term. However, this does not
affect the parameters of the iterative construction of 
almost Ramanujan expanders given in that paper. We thus get an
infinite family of almost Ramanujan quantum expanders i.e.
$(D^n, d, \lambda, 1)$-qTPEs, $n \geq 1$ where $d = 2 s^{\log s}$,
$s \geq 4$ is an even integer, $D = (D_0)^{3 \log s}$, $D_0$
suficiently large integer, and
$\lambda = d^{-\frac{1}{2} + O(\frac{1}{\sqrt{\log d}})}$.

\smallskip

\noindent
2.\ For integer $t$ satisfying $10 t^2 \leq d$,
we thus get an infinite family of almost Ramanujan qTPEs i.e.
$(D^n, d, \lambda_1, t)$-qTPEs where the constraints on the parameters
are the same as in above remark.

\section{Conclusion}
In this paper, we have shown that the famous zigzag product first
defined for classical expander graphs by Reingold, Vadhan and
Wigderson~\cite{RVW} is amazingly powerful: it generalises to
quantum tensor product expanders, and furthermore, it can be
refined via the ideas of Ben-Aroya and Ta-Shma~\cite{BT} to give
efficient constructions of almost Ramanujan $t$-qTPEs for
$t$ polynomial in the number of qubits. This leads to efficient
constructions for unitary $t$-designs for 
$t$ polynomial in the number of qubits. The only efficient construction 
known earlier for such large $t$ was the local random circuit 
construction of Brand\~{a}o, Harrow and Horodecki~\cite{BHH}. 
For both zigzag and generalised zigzag products, our construction has
the advantage of much better tradeoff between the degree and the singular
value gap than what was proved by Brand\~{a}o, Harrow and Horodecki.
For the generalised zigzag product, our tradeoff is almost optimal
by virtue of being almost Ramanujan.
Achieving efficient constructions of perfectly Ramanujan qTPEs remains
an open problem, even for $t=1$. 

Strangely, the zigzag product construction does not seem to work 
for classical tensor product expanders. Finding an efficient combinatorial 
construction of $t$-cTPEs for $t > 1$ is an imporant open problem.
The only efficient constructions known for $t > 1$ are algebraic, 
involving Cayley graphs on the symmetric group~\cite{K}.

\bibliography{expander}

\end{document}